\documentclass[letterpaper,11pt]{article}

\usepackage{amsmath, amsthm, amssymb}
\usepackage[usenames, dvipsnames]{color}
\usepackage[normalem]{ulem} 
\usepackage{fullpage}
\usepackage[numbers, sort]{natbib}
\usepackage[noend]{algorithmic}
\usepackage{enumerate}
\usepackage{hyperref}
\hypersetup{
  hidelinks,
  colorlinks=true,
  allcolors=black,
  pdfstartview=Fit,
  breaklinks=true
}
\usepackage{xspace,xcolor}
\usepackage{placeins}
\usepackage{graphicx}
\usepackage{caption}
\usepackage{subcaption}
\usepackage{cleveref}
\usepackage{enumitem, linegoal}
\usepackage[ruled,vlined]{algorithm2e}
\usepackage{comment}
\newtheorem{claim}{Claim}[section]

\usepackage{mathtools}
\usepackage[flushleft]{threeparttable}
\usepackage{verbatim}
\usepackage{setspace}
\usepackage{bbm}
\usepackage{thm-restate}
\usepackage{footnote}
\makesavenoteenv{tabular}
\makesavenoteenv{table}
\usepackage{authblk}

\usepackage[margin=1in]{geometry}

\definecolor{crimsonglory}{rgb}{0,0,0}

\newtheorem{example}{Example}

 \newtheorem{theorem}{Theorem}[section]
 
 \newtheorem{lemma}[theorem]{Lemma}
 \newtheorem{proposition}[theorem]{Proposition}
 \newtheorem{corollary}[theorem]{Corollary}
 
 \newtheorem{definition}[theorem]{Definition}

\makeatletter
\def\GrabProofArgument[#1]{ #1: \egroup\ignorespaces}
\def\proof{\noindent\textbf\bgroup Proof%
	\@ifnextchar[{\GrabProofArgument}{. \egroup\ignorespaces}}

\makeatother

\usepackage[utf8]{inputenc}
\usepackage{tikz}
\usetikzlibrary{positioning,shapes.geometric,calc}

\newcommand{\etal}{\textit{et al. }}

\newcommand{\apxsplit}{1.76\xspace}

\newcommand{\br}{\boldsymbol{br}}
\newcommand{\bropt}{{\boldsymbol{br}^*}}
\newcommand{\opt}{\boldsymbol{OPT}}
\newcommand{\oh}{\mathcal{O}}

\newcommand{\telebr}{\textsc{Telephone Broadcasting}\xspace}
\newcommand{\twocover}{\textsc{Prefix Covering}\xspace}

\newcommand{\SIMPLE}{\textsc{Basic}\xspace}
\newcommand{\splitalgo}{\textsc{FastSplitBroadcast}\xspace}
\newcommand{\CLASS}{\textsc{Class}\xspace}

\newcommand{\kcycle}{cycle-star\xspace}
\newcommand{\kpath}{melon\xspace}
\newcommand{\Kcycle}{Cycle-star\xspace}
\newcommand{\Kpath}{Melon\xspace}
\newcommand{\roundset}[2]{\mathfrak{R}_{#1}(#2)}

\newcommand{\twocovertwo}{\textsc{Double Prefix Covering}\xspace}

\newcommand{\poly}{\operatorname{poly}}

\input{src/_figure_macros}

\newcounter{proccnt}

\newcommand{\konote}[1]{}

\usepackage[normalem]{ulem} 

\title{On Hardness and Approximation of Broadcasting in Structured Graphs\ \\ \ \\ }

\author[1]{\fontsize{10.5}{12}\selectfont  Jeffrey Bringolf}
\author[1]{\fontsize{10.5}{12}\selectfont Hovhannes A. Harutyunyan\vspace{3mm}}
\author[2]{\ \vspace{1pt}   \\ \fontsize{10.5}{12}\selectfont Shahin Kamali}
\author[2]{\fontsize{10.5}{12}\selectfont Seyed-Mohammad~Seyed-Javadi}

\affil[1]{\fontsize{10.5}{12}\selectfont
Concordia University, Montreal, \vspace{2mm}Canada\footnote{Emails: 
\href{mailto:bringolfj@gmail.com}{bringolfj@gmail.com}, and \href{mailto:hovhannes.harutyunyan@concordia.ca}{hovhannes.harutyunyan@concordia.ca}}}

\affil[2]{\fontsize{10.5}{12}\selectfont
York University, Toronto, Canada\footnote{Emails: 
\href{mailto:kamalis@yorku.ca}{kamalis@yorku.ca}, and \href{mailto:mohammad.sj2@gmail.com}{mohammad.sj2@gmail.com}}}

\SetCommentSty{mycommfont}

\SetKwInput{KwInput}{Input}                
\SetKwInput{KwOutput}{Output}  
\begin{document}
	\newcommand{\ignore}[1]{}
\renewcommand{\theenumi}{(\roman{enumi}).}
\renewcommand{\labelenumi}{\theenumi}
\sloppy
\date{} 
\newenvironment{subproof}[1][\proofname]{
	  \renewcommand{\Box}{ \blacksquare}%
	\begin{proof}[#1]%
	}{%
	\end{proof}%
}
\newenvironment{subproof2}{
	\renewcommand{\Box}{ \blacksquare}%
	\begin{proof}%
	}{%
	\end{proof}%
}

\newlist{pfparts}{description}{1}
\setlist[pfparts,1]{%
  itemindent=2pt,
  wide,
  itemsep=0pt,topsep=2pt,
  labelsep=0.75ex
}


\maketitle

\thispagestyle{empty}
\allowdisplaybreaks

\begin{abstract}
We study the \textsc{Telephone Broadcasting} problem in graphs with restricted structure. Given a designated source in an undirected graph, the goal is to disseminate a message to all vertices in the minimum number of rounds, where in each round every informed vertex may inform at most one uninformed neighbor. For general graphs with $n$ vertices, the problem is NP-hard. Recent work shows that the problem remains NP-hard even on restricted graph classes such as graphs of treewidth 2 [Tale 2025], cactus graphs of pathwidth $2$ [Aminian et~al., ICALP 2025] and graphs at distance-1 to a path forest [Egami et~al., MFCS 2025].

In this work, we investigate the problem in several sparse graph families. We first prove NP-hardness for cycle-star graphs, namely graphs formed by $k$ cycles sharing a single vertex, as well as melon graphs, namely graphs formed by $k$ paths with shared endpoints. Despite multiple efforts to understand the problem in these simple graph families, the computational complexity of the problem remained unsettled. Our hardness results answer open questions by Bhabak and Harutyunyan [CALDAM 2015] and Harutyunyan and Hovhannisyan [COCOA 2023] 
concerning the problem's complexity in cycle-star and melon graphs, respectively. 

On the positive side, we present Efficient Polynomial-Time Approximation Schemes (EPTASs) for cycle-star and melon graphs, improving over the best existing approximation factors of $2$ for both graph families. 
Moreover, we identify a structural frontier for tractability by showing that the problem is solvable in polynomial time on graphs of bounded cutwidth, a class that generalizes other families such as graphs of bounded bandwidth. This result subsumes existing tractability results for special sparse graph families such as necklace graphs.

Finally, for split graphs, a fundamental class of highly structured graphs, we obtain a polynomial-time algorithm with approximation factor \apxsplit. This improves on the previously known factor~$2$ bound; the same approach also applies to the multi-source setting.


\end{abstract}
\newpage
\section{Introduction}

In the \telebr problem, the objective is to spread a message from a designated source vertex to all vertices of a network via a sequence of \emph{telephone calls}. The network is modeled as an undirected, unweighted graph on $n$ vertices, and communication proceeds in synchronous rounds. Initially, only the source knows the message. In each round, every informed vertex may contact at most one uninformed neighbor to pass along the message. The goal is to minimize the total number of rounds until all vertices are informed. 

The problem is known to be computationally difficult. Slater \etal~\cite{slater1981nptree} proved that \telebr is NP-Complete. Elkin and Kortsarz~\cite{elkin2002lowerbound} established that approximating \telebr within a factor of $3-\epsilon$ for any $\epsilon > 0$ is NP-hard. On the other hand, Elkin and Kortsarz~\cite{elkin2006sublogarithmic} designed an $\oh(\log n / \log \log n)$-approximation algorithm for general graphs, which is the best existing approximation algorithm.
Whether a constant-factor approximation exists for general graphs remains open, though such results are known for specific classes, including unit disk graphs~\cite{shang2010unitdisk}, cactus graphs~\cite{aminian2025complexity}, and graphs of bounded pathwidth~\cite{aminian2025complexity,egami2025broadcasting}.

A substantial body of work has investigated the complexity of \telebr on restricted graph classes. On the positive side, polynomial-time algorithms are known for several families, including trees~\cite{fraigniaud2002polynomial}, grids, stars of cliques, fully connected trees, necklace graphs (chains of rings), and certain subclasses of cactus graphs; see, e.g., \cite{damaschke2024starclique, gholami2023fullytree, Liestman1988Boundeddegree, stohr1991butterfly, fomin2024parameterized}. On the negative side, hardness results persist even for structurally simple graphs such as 
planar cubic graphs \cite{jakoby1992minimum}. 
More recently, Tale~\cite{tale2024double} showed that the problem remains NP-hard on graphs at distance~2 from path forests, as well as graphs of pathwidth~3 and treewidth~2.
Aminian \etal \cite{aminian2025complexity} established hardness for so-called \emph{snowflake graphs} (cactus graphs of pathwidth~$2$). 
Even stronger results were obtained by Egami \etal \cite{egami2025broadcasting}, who proved NP-hardness for graphs at distance~1 from path forests, i.e., graphs in which deleting a single vertex yields a path forest. Figure~\ref{fig:combined_figures} illustrates several hard sparse graph instances.

These results highlight intriguing frontier cases where the problem's complexity remains unresolved. One such case is that of \emph{\kcycle graphs} (also called \emph{flower graphs}), obtained by joining $k$ cycles at a single vertex. Despite their simple definition, these graphs exhibit surprisingly rich behavior. \telebr on \kcycle graphs has been studied from an approximation perspective: the best known algorithm achieves a factor of~2~\cite{Bhabak2015}, yet the exact complexity remains open.
%
A closely related family is that of \emph{\kpath graphs}, formed by joining $k$ paths at their endpoints. Several approximation algorithms are known for \kpath graphs~\cite{bhabak2019kpath, HarutyunyanH23}, with the current best result being a 2-approximation~\cite{HarutyunyanH23}. Nevertheless, the precise computational complexity of the problem for this family also remains open. 
Another structurally interesting graph family is that of \emph{split graphs}, which consist of a clique and an independent set and therefore fall outside the sparse regime. For single-source \telebr, a 2-approximation algorithm is known for split graphs~\cite{HarutyunyanHSplit23}. Moreover, the multi-source variant of \telebr is NP-hard even on split graphs~\cite{JansenM95}, underscoring that hardness persists despite strong structural restrictions.

\subsection{Contribution}

Our contributions in this paper are as follows. In a nutshell, we tighten the frontier between the tractable and intractable cases of the \telebr problem. 

\begin{itemize}
    \item We prove that \telebr is NP-Complete in two of the simplest sparse graph families: \emph{\kcycle graphs} (Theorem~\ref{Theorem:HardnessKCylce}) 
    and \emph{\kpath graphs} (Theorem~\ref{Theorem:HardnessKPath}). 
    Both families have been studied in the broadcasting literature, yet their complexity status remained open. Our hardness results also imply hardness for broader classes, including series–parallel graphs (SP-graphs), for which no hardness was previously known. 
    \item We design efficient polynomial-time approximation schemes (EPTASs) for both \kcycle graphs (Corollary~\ref{Corol:PTASkCycle}) and \kpath graphs (Corollary~\ref{Coro:PTASKPath}). To the best of our knowledge, these are the first known graph families for which \telebr admits both a hardness result and a matching EPTAS, highlighting their role as boundary cases between tractability and intractability. Our EPTASs extend naturally to the multi-source setting.

     \item We present a polynomial-time algorithm for optimal broadcasting in graphs of bounded cutwidth (Theorem~\ref{Th:BandWidthMain}). 
This result generalizes several previously known tractable cases, such as necklace graphs (chains of rings) and related structures~\cite{harutyunyan2023chainring}, and also applies to the multi-source setting.
Since 
graphs of bounded bandwidth also have bounded cutwidth,
~\cite{chung1997spectral}, our result further implies tractability for graphs of bounded bandwidth.

     Notably, \telebr is among the few problems exhibiting a complexity gap between graphs of bounded cutwidth (where the problem is tractable) and graphs of bounded pathwidth (where the problem is NP-hard~\cite{tale2024double}). In particular, \kcycle and \kpath graphs both have pathwidth~2 but unbounded cutwidth, placing them precisely on the intractable side of this gap. 

\item We present a $\apxsplit$-approximation algorithm for \telebr on split graphs (Theorem~\ref{Th:SplitMain}), improving over the known factor~2 bound~\cite{HarutyunyanHSplit23}. Together with our tractability results for bounded-cutwidth graphs and hardness results for sparse families, this places split graphs as an intermediate structural class. The algorithm applies to both single-source and multi-source settings, the latter known as being APX-hard~\cite{JansenM95}.


\end{itemize}

\begin{figure}[!t]
    \centering
     \begin{minipage}[b]{0.19\textwidth}
        \centering \hspace{-5mm}
        \scalebox{.29}{\melonFive}
        \subcaption{A graph of distance 2 to path forest~\cite{tale2024double}}
        \label{fig:tate}
    \end{minipage} \  \ 
    \begin{minipage}[b]{0.18\textwidth}
        \centering
        \scalebox{.3}{\snowTwo}
        \subcaption{A snowflake graph (cactus of pathwidth 2)~\cite{aminian2025complexity}}
        \label{fig:snowflake}
    \end{minipage} \  \ 
    \begin{minipage}[b]{0.18\textwidth}
        \centering
        \scalebox{.66}{\theirfig} \vspace{.2cm} \\
        \subcaption{Graph of distance 1 to path forest~\cite{egami2025broadcasting}}
        \label{fig:theirfig}
    \end{minipage} \  \
    \begin{minipage}[b]{0.19\textwidth}
        \centering
        \scalebox{.66}{\kcyclefig} \\ 
        \subcaption{A \kcycle graph  (Theorem~\ref{Theorem:HardnessKCylce})\ \\ \ \\ }
        \label{fig:kcycle}
    \end{minipage} \ \ 
        \begin{minipage}[b]{0.17\textwidth}
        \centering 
        \scalebox{.27}{\kpathgraph} 
        \subcaption{A \kpath graph (Theorem~\ref{Theorem:HardnessKPath})\ \\ }
        \label{fig:kpath}
    \end{minipage}
    
    \caption{Sparse graph families for which \telebr is proven to be NP-hard.}
    \label{fig:combined_figures}
\end{figure}

\subsection{Overview of Techniques}

Our results combine hardness reductions, approximation schemes, and dynamic programming, together with fixed-dimension integer programming and min-cost flow formulations, each tailored to the structural properties of the graph families under consideration. We give a brief outline here.

\begin{itemize}

\item For the negative results, we design reductions that encode scheduling constraints directly into the graph structure, adapting and extending hardness techniques introduced by Tale~\cite{tale2024double} and Egami et al.~\cite{egami2025broadcasting}.
In particular, the NP-hardness proofs for \kcycle and \kpath graphs exploit the interaction between cycle (or path) lengths and the timing of broadcast calls. The shared vertex in \kcycle graphs (or the pair of endpoints in \kpath graphs) acts as the originator, and we show that deciding feasibility becomes equivalent to a combinatorial partitioning problem whose complexity propagates to the \telebr instance. Our reductions build a more restricted variant of \emph{Numerical 3-Dimensional Matching (RN3DM)} (Definition~\ref{def:RN3DM}), previously employed in hardness proofs for \telebr by Tale~\cite{tale2024double} and Egami et al.~\cite{egami2025broadcasting}. The restricted variant that we use, which is known to be still NP-hard~\cite{num-matching-hardness}, makes it possible to encode number-theoretic constraints using only these simple graph families.

\item We design efficient polynomial-time approximation schemes for \kcycle and \kpath graphs by reducing the \telebr problem to covering problems over integers. For \kcycle graphs, the task reduces to covering a multiset of cycle lengths $\{\ell_1,\ldots,\ell_k\}$ with the smallest possible broadcast round $\beta$ using pairs of numbers from $[\,\beta\,]$, such that for each $\ell_i$ there exist $a,b\in[\,\beta\,]$ with $a+b\geq \ell_i$ (see Definition~\ref{def:twocover}). For \kpath graphs, the covering problem is slightly different: we seek the smallest $\beta$ such that each $\ell_i$ can be covered by one number from $[\,\beta\,]$ and one from $[\,\beta-d\,]$, where $d$ is the length of the shortest path between two endpoints (see Definition~\ref{def:twocovertwo}). In both cases, approximate solutions to the resulting number-theoretic covering problems translate into near-optimal broadcasting schedules; moreover, the rounded covering instances have constant support and can be solved exactly using fixed-dimension integer linear programming, yielding EPTASs for \kcycle and \kpath graphs.

\item For finding optimal broadcast schemes in graphs of bounded cutwidth, we exploit the fact that a cutwidth-$k$ ordering restricts each vertex to neighbors within a sliding window of $\oh(k)$ vertices. This allows us to describe the status of any partial broadcast tree using a compact boundary state and use a dynamic programming approach which has a table with $n$ rows (one for each prefix of the ordering) and at most $f(k)\cdot n^{k}$ states per row, giving a total size bounded by $f(k)^2\cdot n^{2k+1}$, in which $f(.)$ is a function of $k$ that is independent of $n$. 

\item 
For split graphs, we establish that informing independent-set vertices early is never beneficial (Lemma~\ref{lemma:goodopt}); equivalently, one can assume without loss of optimality that broadcasts to the independent set are deferred as long as possible.
This structural insight allows us to decouple the broadcast process at each vertex into an initial phase within the clique, followed by a controlled dissemination to the independent set.
As a result, the problem becomes related to constructing minimum-degree and minimum-overload coverings of the independent set by clique vertices. 
When discussing these covering problems, at a high level, we seek to minimize the number of clique vertices that are responsible for covering \emph{many} independent-set vertices. Such \emph{high-load clique vertices} may constitute the main bottleneck of the broadcast process and are therefore informed earlier, so that they can switch to informing independent-set vertices sooner than others. We compute a suitable covering via a reduction to a minimum-cost flow problem, yielding an approximation algorithm with approximation factor~$\apxsplit$ for both the single-source and multi-source settings.

\end{itemize}
\section{Preliminaries}

We use notation $[n]$ and $[n,m]$ for positive integers $n$ and $m$ for the set of all positive integers less than or equal to $n$ and all positive integers between $n$ and $m$ (inclusive), respectively (i.e. $[n]=\{1,2,\ldots, n\}$ and $[n,m]=\{n,n+1,\ldots, m\}$).
Also, we refer to the number of distinct elements in a multiset as its “support”. 
 \begin{definition}[\textsc{Telephone Broadcasting}~\cite{hedetniemi1988broadsurvey}]
    An instance $(G,s)$ of the 
     \emph{\telebr} problem 
     is defined by a connected, undirected, and unweighted graph $G=(V, E)$ and a vertex $s \in V$, where $s$ is the only informed vertex. The broadcasting protocol is synchronous and occurs in discrete rounds. In each round, an informed vertex can inform at most one of its uninformed neighbors. The goal is to broadcast the message as quickly as possible so that all vertices in $V$ get informed in the minimum number of rounds.
 \end{definition}

Given that the problem can be solved in linear time on trees~\cite{fraigniaud2002polynomial}, it is common to represent a \emph{broadcast scheme}—that is, a feasible solution to the problem—by a spanning tree of $G$ rooted at $s$. We refer to such a tree as a \emph{broadcast tree}. In this tree, an edge $(u,v)$ in the tree indicates that vertex $v$ receives the message from $u$ in the broadcast process. Thus, throughout the paper, the terms “broadcast tree’’ and “broadcast scheme’’ are used interchangeably. Likewise, a \emph{broadcast subtree} denotes a partial solution, i.e., a subtree of a broadcast tree capturing how a subset of vertices is informed.
Throughout, we use $\bropt(G,s)$ to denote the optimal broadcast time of graph $G$ when the originator is $s$. When the context is clear, we simply write $\bropt$ in place of $\bropt(G,s)$.

We also consider a multi-source variant of \telebr, in which a set $S \subseteq V$ of vertices is initially informed at round~$0$. The broadcasting process proceeds as before, with each informed vertex informing at most one uninformed neighbor per round. The goal is again to minimize the number of rounds until all vertices are informed. Throughout the paper, unless explicitly stated otherwise, \telebr\ denotes the single-source broadcasting problem.

\begin{definition}[\Kcycle and \Kpath Graphs]\label{def:k-cycle-path} \ 
    A \emph{\kcycle} graph $G$ contains $k$ cycles 
    with exactly one common vertex $r$. 
 A \emph{\kpath} graph $G$ contains $k$ paths with shared endpoints $s$ and $t$. 
\end{definition}

See Figure~\ref{fig:combined_figures} (parts (c) and (d)) for illustrations of \kcycle and \kpath graphs.  
Observe that any \kcycle (respectively, \kpath) graph can be fully described by a set of integers specifying the number of vertices on each cycle (respectively, path).  
This representation allows us to relate such graphs directly to integer sets, a connection that underlies both our hardness proofs and our algorithmic techniques.  
These graph families are also simple representatives of well-studied classes: \kcycle graphs form a subclass of cactus graphs, and \kpath graphs form a subclass of series–parallel graphs.  
Consequently, our hardness results extend beyond \kcycle and \kpath graphs to broader sparse families, including series–parallel graphs, for which no hardness results were previously known.


 \begin{definition}[Graph Cutwidth]
The \emph{cutwidth} of a graph $G=(V,E)$, denoted $\mathrm{cw}(G)$, is the minimum integer $k$ such that there exists a linear ordering $\pi=(v_1,\ldots,v_n)$ of $V$ with the property that, for every $i\in[n]$, at most $k$ edges have one endpoint in $\{v_1,\ldots,v_i\}$ and the other in $\{v_{i+1},\ldots,v_n\}$.
\end{definition}

Computing the cutwidth is known to be fixed-parameter tractable. In particular, for every fixed constant $k$, there exists a linear-time algorithm that either constructs a cutwidth-$k$ vertex ordering or correctly reports that the cutwidth exceeds~$k$~\cite{thilikos2005cutwidth}. This makes cutwidth a particularly well-suited parameter for algorithmic approaches based on linear layouts and dynamic programming, such as those developed in this paper.


\begin{definition}[Split graphs]
A graph $G=(V,E)$ is a \emph{split graph} if its vertex set can be partitioned into two sets $C$ and $I$ such that the subgraph induced by $C$ is a clique and the subgraph induced by $I$ is an independent set.
\end{definition}

\section{Hardness Results}
In this section, we present our hardness results for \telebr in \kcycle graphs and \kpath graphs. These results extend the existing hardness landscape by showing that even some of the simplest graph families exhibit computational intractability. 

Recent hardness results for \telebr typically rely on reductions from variants of the \emph{Numerical Matching with Target Sums} (NMTS) problem, which is known to be strongly NP-hard (Problem SP17 of ~\cite{computers-and-intractability}).
An instance of NMTS consists of three multisets $U, V,$ and $W$, each containing $m$ positive integers. 
The decision problem asks whether there exists a partition of $U \cup V \cup W$ into $m$ disjoint sets 
$A_1, \ldots, A_m$, such that each $A_i$ contains exactly one element $u_i \in U$, one element $v_i \in V$, 
and one element $w_i \in W$, with the property that $u_i + v_i = w_i$ for all $i \in [m]$.

Tale~\cite{tale2024double} introduced a variant of NMTS, called \emph{Numerical 3-Dimensional (Almost) Matching}, in which no repeated entries appear in $U$ and $V$, and all elements of $U \cup V \cup W$ are multiples of $m$ and two integers $T$ and $V$. As in NMTS, the objective is to partition $U \cup V$ into $m$ disjoint sets $A_1, \ldots, A_m$, each containing exactly one element from $U$ and one element from $V$. However, instead of requiring $\sum_{a \in A_i} a = w_i$, the condition is relaxed to 
$\sum_{a \in A_i} a \geq T - \lambda.$
Tale showed that this problem is strongly NP-hard via a reduction from NMTS, and subsequently reduced it to \telebr in graphs of pathwidth at most~3. 
Egami et al.~\cite{egami2025broadcasting} studied another restricted variant of NMTS, in which all members of $U$ are even while all members of $V \cup W$ are odd. They proved that this variant remains strongly NP-hard and provided a reduction to \telebr in graphs that are distance-1 from a path forest. 

To establish our hardness results, we rely on another restricted version of NMTS in which $U = V = [m]$. This problem, referred to as \emph{Restricted Numerical 3-Dimensional Matching (RN3DM)}, was studied by Yu et al.~\cite{num-matching-hardness}, who established its strong NP-hardness. Their proof relies on a reduction from the \emph{two-machine flowshop scheduling problem with delays}, which they first showed to be NP-hard. Formally, RN3DM is defined as follows.

\begin{definition}[\textsc{Restricted Numerical 3-Dimensional Matching (RN3DM)}~\cite{num-matching-hardness}]
Given a multiset $W = \{w_1, w_2, \ldots, w_m\}$ of positive integers such that
$\sum_{i=1}^m w_i + m(m+1) = m\cdot e$ for some integer $e$, 
do there exist two permutations $\lambda$ and $\mu$ of $[m]$ such that
$\lambda(i) + \mu(i) + w_i = e \quad \text{for all } 1 \leq i \leq m \, $?
\label{def:RN3DM}
\end{definition}



\begin{example}[Yes-instance of RN3DM]\label{example:RN3DM}
Consider the instance of RN3DM with $W=\{1,3,4,4\}$, so $m=4$. Note that we have $\sum_{i=1}^m w_i + m(m+1) = 
32 = m\cdot e$ for $e=8$.
Let $\lambda = (4,1,3,2)$ and $\mu = (3,4,1,2)$. One can check that for every $i \in [4]$, $\lambda(i) + \mu(i) + w_i = e.$
For instance, $\lambda(1)+\mu(1)+w_1 = 4+3+1=8$, and $\lambda(2)+\mu(2)+w_2 = 1+4+3=8.$
The same holds for $i=3,4$. Hence $\lambda$ and $\mu$ certify that $W$ is a \emph{yes-instance} of RN3DM.
\end{example}

Given that RN3DM is strongly NP-hard~\cite{num-matching-hardness}, we assume that the numbers in $W$ are encoded in unary. 
We present reductions from RN3DM to \telebr in \kcycle and \kpath graphs. 
For \kcycle graphs, our approach introduces an intermediate problem as part of the reduction series, 
whereas for \kpath graphs, we directly reduce RN3DM to \telebr.

\subsection{Hardness of Telephone Broadcasting in \Kcycle Graphs}

We prove that the \kcycle \telebr is NP-hard through a series of reductions starting from RN3DM. 
As a first step, we define an intermediate problem, called \emph{Even-Odd RN3DM}, which is a variant of RN3DM. 
In this variant, instead of using two permutations $\lambda$ and $\mu$ of $[m]$ (matched with the input set $W$), we use $\alpha$ as a permutation of the even integers in $[2m]$ and $\beta$ as a permutation of the odd integers in $[2m]$ to be matched with an input set that we call $C$. 
Formally, Even-Odd RN3DM is defined as follows.

\begin{definition}[\textsc{Even-Odd Restricted Numerical 3-Dimensional Matching (Even-Odd RN3DM)}]
Given a multiset $C = \{c_1, c_2, \ldots, c_m\}$ of positive integers such that 
$\sum_{i=1}^m c_i = m(2m+1),$
does there exist a permutation $\alpha$ of $\{2,4,6,\ldots,2m\}$ and a permutation $\beta$ of $\{1,3,5,\ldots,2m-1\}$ such that 
$
\alpha(i) + \beta(i) = c_i \quad \text{for all } 1 \leq i \leq m \, ?
$
\end{definition}

Observe that any instance of Even-Odd RN3DM in which $C$ contains an even element is trivially a \textsc{No}-instance, since each $c_i$ must be represented as the sum of one even and one odd number, which always yields an odd value. Hence, without loss of generality, we restrict attention to instances where all elements of $C$ are odd.

\begin{example}[Yes-instance of Even-Odd RN3DM]
\label{example:EvenOddRN3DM}
Consider the instance of Even-Odd RN3DM with $C=\{7,7,9,13\}$, where $m=4$. 
Note that $\sum_{i=1}^4 c_i = 36 = m(2m+1)$. 
Now take the permutations $\alpha = (4,6,2,8)$ and $\beta = (3,1,7,5)$. 
One can check for every $i\in[4]$, $\alpha(i)+ \beta(i) = c_i$. For instance
$\alpha(1)+\beta(1) = 4+3 =  c_1$, and
$\alpha(2)+\beta(2) = 6+1 = c_2$. 
The same holds for $i=3,4$. Thus, $\alpha$ and $\beta$ certify that $C$ is a \emph{yes-instance} of Even-Odd RN3DM. \vspace{-2mm}
\end{example}

\newcommand{\evenOddLemma}[1]{
\begin{lemma}\label{Lem:EvenOddRN3DM}
\emph{#1}Even-Odd RN3DM is strongly NP-hard. 
\end{lemma}}

\begin{figure}
    \centering
    \scalebox{.46}{\kcycleReductionNew}
    \caption{An instance $C=\{7,7,9,13\}$ of Even-Odd RN3DM reduces to the above instance of \telebr in \kcycle graphs (with $s$ as the originator and $br = 8$). The broadcast tree associated with solution $\alpha = (4,6,2,8)$ and $\beta = (3,1,7,5)$ is highlighted. 
}
    \label{fig:kcycleReduction}
\end{figure}

\evenOddLemma{}

\begin{proof}[sketch]
We use a reduction from RN3DM, which is strongly NP-hard~\cite{num-matching-hardness}. 
Given an instance $W=\{w_1,\ldots,w_m\}$ of RN3DM with target sum $e$, we construct an instance $C=\{c_1,\ldots,c_m\}$ of Even-Odd RN3DM by setting
$c_i = 2e - 2w_i - 1 \quad \text{for each } i \in [m].$
It is straightforward to verify that $W$ is a \textsc{Yes}-instance of RN3DM if and only if $C$ is a \textsc{Yes}-instance of Even-Odd RN3DM. 
\end{proof}

\begin{proof}
We give a polynomial-time reduction from RN3DM. 
Let $W=\{w_1,\ldots,w_m\}$ be an arbitrary instance of RN3DM. 
By definition, we have $\sum_{i=1}^m w_i + m(m+1) = me$ for some integer $e$. Construct an instance $C=\{c_1,\ldots,c_m\}$ of Even-Odd RN3DM, where  
$c_i = 2e - 2w_i - 1 \quad \text{for each } i \in [m].$  
To illustrate, the instance of RN3DM in Example~\ref{example:RN3DM} reduces to the instance of Even-Odd RN3DM in Example~\ref{example:EvenOddRN3DM}.
First, we verify that $C$ satisfies the side condition of Even-Odd RN3DM:
\[ 
\sum_{i=1}^m c_i 
= \sum_{i=1}^m (2e - 2w_i - 1) 
= 2em - 2(\sum_{i=1}^m w_i) - m 
= 2em - 2(me - m(m+1)) - m 
= m(2m+1).
\]
In the above equalities, we used the definition of $C$ and the side condition that holds for the instance $W$ of RN3DM. 
Thus, $C$ is a valid instance of Even-Odd RN3DM.  

Next, we show the correctness of the reduction.  

\smallskip
\noindent \emph{(Forward direction).}  
Suppose $W$ is a yes-instance of RN3DM. Then there exist permutations $\lambda$ and $\mu$ of $[m]$ such that for each $i \in [m]$, $\lambda(i) + \mu(i) + w_i = e.$
Define $\alpha$ as the permutation of $\{2,4,\ldots,2m\}$ given by $\alpha(i) = 2\lambda(i)$, and $\beta$ as the permutation of $\{1,3,\ldots,2m-1\}$ given by $\beta(i) = 2\mu(i)-1$. Then, for every $i$,
\[
\alpha(i) + \beta(i) = 2\lambda(i) + (2\mu(i)-1) = 2(\lambda(i)+\mu(i)) - 1 = 2(e-w_i) - 1 = 2e - 2w_i - 1 = c_i.
\]
Hence, $\alpha$ and $\beta$ certify that $C$ is a yes-instance of Even-Odd RN3DM.  

\smallskip
\noindent \emph{(Reverse direction).}  
Suppose $C$ is a yes-instance of Even-Odd RN3DM. Then there exist permutations $\alpha$ of $\{2,4,\ldots,2m\}$ and $\beta$ of $\{1,3,\ldots,2m-1\}$ such that for each $i \in [m]$, $\alpha(i) + \beta(i) = c_i.$ 
Define $\lambda(i) = \alpha(i)/2$ and $\mu(i) = (\beta(i)+1)/2$. Both $\lambda$ and $\mu$ are permutations of $[m]$. For each $i$,
\[
\lambda(i) + \mu(i) + w_i 
= \frac{\alpha(i)}{2} + \frac{\beta(i)+1}{2} + w_i 
= \frac{c_i+1}{2} + w_i 
= \frac{(2e - 2w_i - 1)+1}{2} + w_i 
= e.
\]
Thus, $\lambda$ and $\mu$ certify that $W$ is a yes-instance of RN3DM.  

\smallskip
Since the reduction preserves yes-instances in both directions and can be carried out in polynomial time, Even-Odd RN3DM is strongly NP-hard.
\end{proof}

We are ready to prove the main result of this section, namely, the hardness of \telebr in \kcycle.

\begin{theorem}
\telebr in \kcycle graphs is NP-Complete. \label{Theorem:HardnessKCylce}
\end{theorem}

\begin{proof}
    \telebr is in NP in general graphs (and hence also for \kcycle graphs). 
    We prove NP-hardness by a reduction from the 
    Even-Odd RN3DM problem, which is strongly NP-hard by Lemma~\ref{Lem:EvenOddRN3DM}. 
    
    Let $C= \{c_1,c_2,\ldots,c_m\}$ be an instance of Even-Odd RN3DM. 
    We construct an instance $(G,s,t)$ of \telebr in \kcycle graphs as follows. 
    Graph $G$ consists of $m$ cycles, containing $c_1, c_2, \ldots, c_m$, extra vertices in addition to a single vertex $s$ that is shared between them. Now $s$ serves as the originator of the broadcast. 
    Thus, the number of vertices in $G$ is $n = \sum_{i=1}^m c_i + 1 = m(2m+1) + 1,$
    and the degree of $s$ is $2m$. 
    Since RN3DM is strongly NP-hard by Lemma~\ref{Lem:EvenOddRN3DM}, we may assume the values $c_i$ are encoded in unary, so the reduction is polynomial time. 
    The decision problem asks whether broadcasting completes within $t = 2m$ rounds. See Figure~\ref{fig:kcycleReduction}.  
    
\smallskip
\noindent \emph{(Forward direction).}  
 Suppose the answer to instance $C$ of Even-Odd RN3DM is yes. 
    Then there exist permutations $\alpha$ of $\{2,4,\ldots, 2m\}$ and $\beta$ of $\{1,3,\ldots, 2m-1\}$ such that $\alpha(i)+\beta(i) = c_i$ for all $i \in [m]$. 
    Construct a broadcast scheme as follows: 
    the source $s$ informs its two neighbors in each cycle $c_i$ at times $2m-\alpha(i)+1$ and $2m-\beta(i)+1$, respectively. Given that $\alpha$ and $\beta$ are permutations of space $[2m]$, the source $s$ informs exactly one vertex at each given round $i\in[2m]$. 
    Each newly informed vertex continues by informing its uninformed neighbor in the subsequent round. 
    Hence, by round $br=2m$, exactly $\alpha(i)+\beta(i)=c_i$ vertices of the cycle (excluding $s$) are informed, and, therefore, the entire graph is informed. 
    Thus, $(G,s,br)$ is a yes-instance of \telebr. 

\smallskip
\noindent \emph{(Reverse direction).}  
 Now suppose $(G,s,br)$ is a yes-instance of \telebr with $br=2m$. Therefore, there is a broadcast scheme that completes within $2m$ rounds.
    Without loss of generality, assume the scheme is non-lazy, i.e., every informed vertex calls an uninformed neighbor whenever possible. 
    Let $t_i$ denote the number of vertices newly informed in round $i$. 
    At round 0, only $s$ is informed: $t_0=1$. 
    At each round $i$, only $s$ and the vertices informed at time $i-1$ can inform new vertices (since other vertices of degree 2 have no additional uninformed neighbors left). 
    Hence, $t_{i+1} \leq t_i + 1$. 
    It follows that the number of informed vertices by round $i$ is at most 
    $        1 + 1 + 2 + \cdots + i = 1 + \frac{i(i+1)}{2}$.
    For $i=2m$, this reaches exactly $
        1 + m(2m+1) = n$,
    so the broadcast must finish in exactly $2m$ rounds (and not earlier). 
    This forces the source $s$ to make a call in every round, and since $\deg(s)=2m$, $s$ informs exactly one neighbor per round. 
    Consequently, each cycle $c_i$ must be partitioned into two paths, informed from $s$ in opposite directions, of lengths $\alpha(i)$ and $\beta(i)$ with $\alpha(i)+\beta(i)=c_i$. 
    Because $c_i$ is odd (otherwise the Even-Odd RN3DM instance would be trivial), we conclude that one of $\alpha(i),\beta(i)$ is odd and the other is even. 
    This provides the exact decomposition required for a yes-instance of Even-Odd RN3DM. 
%
 %
\end{proof}

\subsection{Hardness of Telephone Broadcasting in \Kpath Graphs}

Next, we use a reduction from RN3DM to establish NP-Completeness of \telebr in \kpath graphs.

\newcommand{\melonHardnessTheorem}[1]{
\begin{theorem}
\emph{#1} \telebr in \kpath graphs is NP-Complete. \label{Theorem:HardnessKPath}
\end{theorem}}

\melonHardnessTheorem{}

\begin{proof}[sketch]
We prove NP-hardness by a reduction from the \emph{Restricted Numerical 3-Dimensional Matching} (RN3DM) problem (Definition~\ref{def:RN3DM}), which is strongly NP-hard~\cite{num-matching-hardness}.

Let $W=\{w_1,w_2,\ldots,w_m\}$ be an instance of RN3DM. 
By definition, $\sum_{i=1}^m w_i + m(m+1) = m\cdot e$ for some integer $e$. 
We construct an instance $(G,s,br)$ of \telebr in \kpath graphs as follows.    
The graph $G$ consists of $m$ internally disjoint paths with shared endpoints $s$ and $t$, where the $i$-th path contains 
$\ell_i = e - w_i$
internal vertices in addition to $s$ and $t$. Furthermore, we add an edge connecting $s$ and $t$. 
The vertex $s$ serves as the source of the broadcast.  
The total number of vertices in $G$ is $n = 2 + \sum_{i=1}^m (e - w_i) = 2 + me - \sum_{i=1}^m w_i = 2 + m(m+1).$
Both $s$ and $t$ have degree $2m+1$.  
Since RN3DM is strongly NP-hard, we may assume that the values $w_i$ are encoded in unary, ensuring that the reduction is polynomial time.  
The decision problem asks whether broadcasting completes within $br = m+1$ rounds (see Figure~\ref{fig:pathReduction}). The answer to this decision problem is yes if and only if the $W$ is a yes-instance of RN3DM. 
\end{proof}

\begin{proof}
\telebr belongs to NP for general graphs and hence also for \kpath graphs. 
We prove NP-hardness by a reduction from the \emph{Restricted Numerical 3-Dimensional Matching} (RN3DM) problem (Definition~\ref{def:RN3DM}), which is strongly NP-hard~\cite{num-matching-hardness}.

Let $W=\{w_1,w_2,\ldots,w_m\}$ be an instance of RN3DM. 
By definition, $\sum_{i=1}^m w_i + m(m+1) = m\cdot e$ for some integer $e$. 
We construct an instance $(G,s,br)$ of \telebr in \kpath graphs as follows.    
The graph $G$ consists of $m$ internally disjoint paths with shared endpoints $s$ and $t$, where the $i$-th path contains 
$\ell_i = e - w_i$
internal vertices in addition to $s$ and $t$. Furthermore, we add an edge connecting $s$ and $t$. 
The vertex $s$ serves as the source of the broadcast.  
The total number of vertices in $G$ is $n = 2 + \sum_{i=1}^m (e - w_i) = 2 + me - \sum_{i=1}^m w_i = 2 + m(m+1).$
Both $s$ and $t$ have degree $2m+1$.  
Since RN3DM is strongly NP-hard, we may assume that the values $w_i$ are encoded in unary, ensuring that the reduction is polynomial time.  
The decision problem asks whether broadcasting completes within $br = m+1$ rounds (see Figure~\ref{fig:pathReduction}).  

\smallskip
\noindent \emph{(Forward direction).}  
Suppose $W$ is a yes-instance of RN3DM. 
Then there exist permutations $\lambda$ and $\mu$ of $[m]$ such that $\lambda(i)+\mu(i) = e-w_i = \ell_i$ for all $i \in [m]$.  
Construct a broadcast schedule as follows:  
in round~1, the source $s$ calls $t$.  
Subsequently, in round $j \in [m]$, $s$ and $t$ inform their neighbors in the $j$'th path at times $m+1-\lambda(j)$ and $m+1-\mu(j)$, respectively.  
Since $\lambda$ and $\mu$ are permutations, exactly one new vertex is informed from $s$ and one from $t$ in each round (i.e., this is a valid broadcast scheme).  
Each newly informed vertex continues the broadcast by informing its remaining neighbor in the next round.  
Hence, by round $br=m+1$, exactly $\lambda(j)$ vertices from the side of $s$ and $\mu(j)$ vertices from the side of $t$ from path $j$ are informed; therefore, all $\ell_j$ internal vertices of every path are informed, and the entire graph is covered.  
Thus, $(G,s,br)$ is a yes-instance of \telebr.  

\smallskip
\noindent \emph{(Reverse direction).}  
Suppose $(G,s,br)$ is a yes-instance of \telebr with $br = m+1$. 
Without loss of generality, assume the scheme is \emph{non-lazy}, i.e., every informed vertex calls an uninformed neighbor whenever possible.  
Let $t_i$ denote the number of vertices newly informed in round $i$.  
At round~0, only $s$ is informed, so $t_0=1$.  
Suppose $t$ is informed at round $p \geq 1$.  
At each round $i$, only $s$, possibly $t$ (if $i>p$), and the vertices informed at round $i-1$ can make new calls.  
Thus $t_{i+1} \leq t_i + 1$ for $i \leq p$, and $t_{i+1} \leq t_i + 2$ for $i \geq p$.  
This upper bound value is maximized when $p=1$, i.e., when $s$ calls $t$ immediately.  
In that case, we have $t_0=1$, $t_1=1$, $t_2=2$, and for $i \geq 2$, $t_i = t_{i-1}+2$.  
Therefore, the total number of informed vertices by round $i$ is at most $1 + 1 + 2 + 4 + 6 + \cdots + 2(i-1) = i(i-1) + 2.$
For $i=m+1$, this reaches exactly $n = 2+m(m+1)$.  
Hence, the broadcast cannot finish in less than $m+1$ rounds, i.e., must finish in exactly $m+1$ rounds.  
This forces $s$ to call $t$ in round~1, and both $s$ and $t$ to call exactly one new neighbor in each subsequent round.  Otherwise, the number of informed vertices at time $t_i$ will be less than the above upper bound, and hence the broadcast cannot complete within $m+1$ rounds.

Consequently, each path of length $\ell_i$ must be partitioned into two subpaths, one informed from $s$ and the other from $t$, of lengths $\lambda(i)$ and $\mu(i)$ such that $\lambda(i)+\mu(i)=\ell_i$. Recall that $\ell_i = e-w_i$. Since both $s$ and $t$ inform exactly one new neighbor in each round, it follows that for any two distinct paths $i$ and $j$, we must have $\lambda(i) \neq \lambda(j)$ and $\mu(i) \neq \mu(j)$. In other words, $\lambda$ and $\mu$ are permutations of $[m]$. This corresponds exactly to a yes-instance of RN3DM.
\end{proof}

\begin{figure}
    \centering
    \scalebox{.43}{\melonNew} \vspace{-1mm}
    \caption{An instance $W=\{1,3,4,4\}$ of RN3DM reduces to the above instance of \telebr in \kcycle graphs (with $s$ as the originator and $br = 5$). The broadcast tree associated with with solution $\lambda = (4,1,3,2)$ and $\mu = (3,4,1,2)$ is highlighted. 
}
    \label{fig:pathReduction}
\end{figure}
\section{Efficient Polynomial Time Approximation Schemes}
\subsection{EPTAS for \Kcycle Graphs}
In this section, we present an efficient polynomial-time approximation scheme (EPTAS) for \telebr in \kcycle graphs.  
Recall that a \kcycle graph is obtained by joining $k$ cycles at a common central vertex $r$.  We assume that the number of cycles $k$—and hence the optimal broadcast time—is asymptotically large; for constant $k$, the problem can be solved optimally by a straightforward exhaustive search (see, e.g.,~\cite{vcevnik2017broadcasting}).


In what follows, we design an EPTAS for the case where the broadcast originates at $s=r$. 
This restriction is without loss of generality: any $c$-approximation for the case that the center is the originator (when $s=r$) can be extended to a $c$-approximation in the general case ($s\neq r$). Suppose the originator $s$ is a non-center vertex on some cycle $Y$, at distance $d \geq 1$ from $r$. The message can first be forwarded along the shortest path from $s$ to $r$. Once $r$ is informed, it first sends the message to its other neighbor on $Y$ and then proceeds with the $c$-approximation algorithm on the rest of the graph. If the overall broadcast time is determined by a vertex in $Y$, then the scheme is optimal. Otherwise, the process completes within $d + 1 + c \cdot \boldsymbol{br}^-$ rounds, where $\boldsymbol{br}^-$ is the optimal broadcast time excluding $Y$. Since $\bropt \geq d + \boldsymbol{br}^-$, this yields the same approximation factor $c$.

\paragraph*{From \telebr in \Kcycle Graphs to \twocover}

We introduce the following problem, which is the key part of our EPTAS algorithms. 

\begin{definition}\label{def:twocover}
The input to the \twocover problem is a \emph{ground multiset} $S=\{\ell_1,\ell_2,\ldots,\ell_n\}$ of positive integers. The objective is to determine the minimum integer $m$ such that the set $C=[m]$ can be partitioned into $n+1$ subsets $C_0, C_1,\ldots,C_n$, such that 
    for each $i\in[n]$, we have $|C_i|\leq 2$  and
 $\sum_{x\in C_i} x \,\geq\, \ell_i$. 
In this context, we say that $\ell_i$ is \emph{covered} by the elements of $C_i$, and we refer to $C$ as the \emph{covering set} of $S$.

\end{definition}

\begin{example}
When $S=\{13,8,7,6\}$, the optimal solution is given by $C=[8]$, and we have $C_1 = \{8,5\}, C_2 =\{7,2\}, C_3 = \{4,3\},  C_4 =\{6\}$ (and $C_0 = \{1\}$).    
\end{example}





We begin by establishing the equivalence between \telebr in \kcycle graphs and the \twocover problem. 


\newcommand{\kcyclePrefixEquivalence}[1]{
\begin{lemma}\label{lemma:kcycle-equal-twocover}
\emph{#1} Let $G$ be a \kcycle graph with central vertex $s$, where the $i$-th cycle contains $c_i$ additional vertices besides $s$, and let $S=\{c_1,\ldots,c_k\}$.  
Then \telebr for instance $(G,s)$ becomes equivalent to \twocover for $S$.  
\end{lemma}
}

\kcyclePrefixEquivalence{}

\begin{proof}[sketch]
Let $\br$ denote the number of broadcast rounds.  
If the center $s$ informs two neighbors of cycle $i$ at rounds $x$ and $y$, then the remaining $c_i$ vertices are informed within $\br$ rounds if and only if we have $(\br-x+1)+(\br-y+1)\geq c_i$.  
This is equivalent to covering $c_i$ by two elements of $[\,\br\,]$ whose sum is at least $c_i$.  
Thus, a broadcast schedule with $\br$ rounds corresponds to a feasible \twocover solution, and conversely any such covering yields a valid broadcast schedule. 
\end{proof}

\begin{proof}
We show that any feasible broadcasting scheme with $\br$ rounds, for instance $(G,s)$, corresponds to a feasible covering of $S$ using $[\br]$, and conversely any such covering yields a valid broadcasting scheme. 
This one-to-one correspondence ensures that the two problems have the same optimal solutions.

\medskip\noindent
\emph{From \telebr to \twocover:}  
In any broadcast scheme with $\br$ rounds, $r(=s)$ must call at most two vertices on each cycle.  
If $r$ informs neighbors in cycle $i$ at rounds $x$ and $y$, then the remaining $c_i-1$ vertices of cycle $i$ must be covered within the remaining rounds, which is feasible exactly when $(\br-x+1)+(\br-y+1)\geq c_i$.  
Thus, each $c_i$ is covered by the pair $\{\br-x+1,\br-y+1\}$ in the set $[\br]$.  
Since $r$ uses distinct rounds for distinct neighbors, the covering also respects the injectivity requirement.  
Hence, every broadcast scheme with $\br$ rounds yields a valid covering for $S$ with $[\br]$.  

\medskip\noindent
\emph{From \twocover to \telebr:}  
Conversely, suppose $S$ has a covering using $C=[\br]$.  
If $c_i$ is covered by $x$ and $y$, then in the broadcast scheme $r$ informs two neighbors in the $i$'th cycle at rounds $\br-x+1$ and $\br-y+1$.  
Because $x+y \geq c_i$, the message traverses all of $i$'th cycle within $\br$ rounds.  
Injectivity of the covering ensures that $r$ never attempts more than one call in a round.  
Thus, every covering for the \twocover instance corresponds to a valid broadcast scheme with $\br$ rounds.  

Therefore, \telebr in \kcycle graphs is equivalent to the \twocover problem. 
\end{proof}



\paragraph*{EPTAS for \twocover}\label{subsec:twocover}
In this section, we present an EPTAS for the \twocover problem, that is, an algorithm with approximation factor $(1+\epsilon)$ for a given $\epsilon > 0$. By Lemma~\ref{lemma:kcycle-equal-twocover}, this would give an EPTAS for \telebr in \kcycle graphs. 
As a high-level view of the algorithm, it addresses the covering problem through a two-stage approximation scheme that leverages the concept of rounded sets (see \Cref{def:roundset}) to reduce the problem complexity. The main idea is to transform the original ground multiset and the covering set into a simplified version with constant support (i.e., distinct elements), enabling formulation as a feasibility ILP with a constant number of variables. Solving this ILP  
yields to the desired EPTAS. 

The algorithm operates as follows.  
Given an instance $S$ of \twocover, we first construct a \emph{rounded set} $\roundset{p}{S}$, where each element of $S$ is rounded up to a larger value in $S$ so that the support of $\roundset{p}{S}$ has size $p$.  
Here $p$ is a constant integer chosen as a function of $\epsilon$. Since rounding increases item values, covering $\roundset{p}{S}$ may require a larger covering set.  
Nevertheless, we show that the optimal covering of $\roundset{p}{S}$ is within a factor $(1+2/p)$ of the optimal covering of $S$ (\Cref{proposition:optroundset}).  
Thus, it suffices to focus on covering $\roundset{p}{S}$ rather than $S$.  
For that, we use a binary search over candidate values $m$, where $m$ specifies the size of the covering set $C=[m]$.  
For each $m$, we define a rounded version $\roundset{p}{C}$, obtained by rounding up the elements of $C$ so that its support also has size $p$.  
We prove that if $\roundset{p}{S}$ can be covered using $\roundset{p}{C}$, then it can also be covered using $C'=[m']$ with $m' \leq (1+1/p)m$ (\Cref{claim:roundsetcover}).  
Conversely, if $\roundset{p}{S}$ cannot be covered by $\roundset{p}{C}$, then it cannot be covered by $C (=[m])$.  
Hence, verifying whether $\roundset{p}{C}$ successfully covers $\roundset{p}{S}$ 
is sufficient for approximating the original instance.  
This verification can be performed by solving an ILP over $\roundset{p}{S}$ and $\roundset{p}{C}$.  
Since both sets have constant support (bounded by $p$), this step can be carried out in time independent of $n$ (\Cref{lemma:exhaustive}).  
In what follows, we provide the details of the algorithm.



Based on \Cref{def:twocover}, an instance of \twocover can be defined as a tuple $(S, C)$, where $S=\{\ell_1,\ell_2,\ldots,\ell_n\}$ and $C = [m]$ are the ground multiset and the covering set, respectively. 
Throughout this section, we assume that $\ell_1 \geq \ell_2\geq \ldots \geq \ell_n$.  

\begin{definition}\label{def:roundset}
    Let $S$ be a finite set and let $p$ be a positive integer. The \emph{rounded set} $\roundset{p}{S}$ of $S$ with $p$ parts is constructed as follows:
    \begin{enumerate}
        \item Order the elements of $S$ in non-increasing sequence: $s_1 \geq s_2 \geq \cdots \geq s_{|S|}$.
        
        \item Partition this ordered sequence into $p$ consecutive parts $P_1, P_2, \ldots, P_p$, where:
        \begin{itemize}
            \item For $i = 1, 2, \ldots, p-1$: $P_i$ contains exactly $\lceil |S|/p \rceil$ elements,
            \item $P_p$ contains the remaining $|S| - (p-1)\lceil |S|/p \rceil$ elements.
        \end{itemize}
        
        \item For each part $P_i$, let $m_i$ denote the maximum element in $P_i$. Then $\roundset{p}{S}$ is the multiset obtained by replacing each element in $P_i$ with $m_i$ for all $i \in \{1, 2, \ldots, p\}$.
    \end{enumerate}
\end{definition}

\begin{example}\label{ex1}
    For
    $S = 
    \{123, 67, 65, 45, 
    43, 43, 43, 18, 
    12, 12, 10, 6, 
    4, 4, 1, 1\}$ and $p=4$, 
    we get
    $\roundset{4}{S} = 
    \{123, 123, 123, 123, 
    43, 43, 43, 43, 
    12, 12, 12, 12, 
    4, 4, 4, 4\}$. 
\end{example}

\begin{proposition}\label{proposition:optroundset}
    $\opt(\roundset{p}{S}) \leq \opt(S) + 2\lceil |S|/p\rceil$, where $\opt$ is the optimal cover for a specific instance of \twocover.
\end{proposition}
\begin{proof}
    By definition, $\roundset{p}{S}$ consists of $p$ parts.
    Given an optimal cover of $S$, form a cover of $\roundset{p}{S}$ by assigning the numbers given to the part $i$ in $S$ to the part $(i+1)$ in $\roundset{p}{S}$ for $i\in[p-1]$. As the size and value of part $i+1$ is not greater than the value of part $i$, the set of numbers that can cover part $i$ can also cover part $i+1$. 

    Therefore, the first part is the only part that needs to be covered. We show that adding an extra $2\lceil |S|/p \rceil$ to the covering set is enough to cover the first part. This is because every $\ell_i\in S$ is less than $2$ items of size $\opt(S)$, as otherwise it would not be possible to cover $\ell_i$ with at most two numbers in $[\opt(S)]$. Thus, by adding an extra $2 \lceil |S|/p \rceil$ items to the covering set, we can cover $\roundset{p}{S}$.
\end{proof}
\begin{example}[Based on \Cref{ex1}]
    In
    \Cref{ex1}, items in a covering set that were used to cover $\{123, 67, 65, 45\}$ in the optimal cover of $S$ will be used to cover $\{43, 43, 43, 18\}$ in $\roundset{p}{S}$; those used to cover $\{43,43,43,18\}$ will be used to cover $\{12,12,12,12\}$ in $\roundset{p}{S}$, and so on. The only remaining uncovered part in $\roundset{p}{S}$ will be the first part ($(123, 123, 123, 123)$), which has size $|S|/p$. 
    Moreover, every element in that part has size at most $2\opt(S)$. 
    Therefore, the first part in $\roundset{p}{S}$ can be covered by an extra $2|S|/p$ in the covering set. In the example, if $\opt(S) = 73$ in the above example, we can cover $\{123, 123, 123, 123\}$ by adding $\{74, 75, 76, 77, 78, 79, 80, 81\}$ to the optimal cover set of $S$ to cover $\roundset{p}{S}$ using $\opt(S) + 2|S|/p$ rounds. 
\end{example}

Given that $\opt(S) \geq |S|$, we can conclude that an optimal covering of $\roundset{p}{S}$ gives a ($1+2/p$) approximation for covering $S$. Therefore, we can now focus on covering $\roundset{p}{S}$, which has constant support $p$.
 

Let $m$ be a given value and let $C=[m]$.  
We present a procedure that either constructs a covering of $\roundset{p}{S}$ using the set $C' = [m']$, where $m' = m + \lceil m/p \rceil$, or certifies that $\roundset{p}{S}$ cannot be covered using $[m]$.  
For this purpose, we use the rounded set $\roundset{p}{C}$ to cover $\roundset{p}{S}$.  
We define the \emph{stretched covering} $\roundset{p}{C}$ as the set obtained by partitioning the $C$ and rounding up its elements (as in Definition~\ref{def:roundset}), ensuring that the covering set has constant support $p$.  
This enables a binary-search strategy to find the smallest feasible value of $m$.

\begin{example}\label{ex2}
    Suppose 
    $\roundset{p}{S} = \{123, 123, 123, 123, 43, 43, 43, 43, 12, 12, 12, 12, 4, 4, 4, 4\}$ and $p=3$.  
    For $m = 12$, we have
    $C= \{12, 11, 10, 9, 8, 7, 6, 5, 4, 3, 2, 1\}$, and we aim to cover $\roundset{p}{S}$ with
    $\roundset{p}{C} = \{12, 12, 12, 12, 8, 8, 8, 8, 4, 4, 4, 4\}$.
\end{example}

\begin{claim}\label{claim:roundsetcover}
    Given that we can cover $S$ with $\roundset{p}{C}$, where $C=[m]$, we can cover $S$ using $C'=[m']$, where $m'=m+\lceil m/p \rceil$.
\end{claim}
\begin{proof}
    Based on definition, $\roundset{p}{C}$ contains $p$ parts. 
    Suppose we have a covering with $\roundset{p}{C}$ for $S$.
    Associate value $m+\lceil m/p \rceil - (i-1)$ to the $i$'th largest number in $C$ for $i\in[m]$. This transformation only increases the value as the first block with original value $m$ matches to $m+1, \ldots, m+\lceil m/p \rceil$, the second block matches to $m-\lceil m/p \rceil +1 ,\ldots, m$, and so on. More specifically, the $i$'th block matches to $m-(i-1)\lceil m/p \rceil +1 ,\ldots, m-(i-2)\lceil m/p \rceil$, which is greater than the original value of the block. 
    Now, as the new set is a subset of $C'$ and every value in $C$ is larger than its original value, the same covering can cover $S$ with the corresponding values in $C'$. 
\end{proof}
\begin{example}[Based on \Cref{ex2}]
    If we can cover $\roundset{p}{S}$ using
    $\roundset{p}{C} = \{ 12, 12, 12, 12, 8,$ $ 8, 8, 8, 4, 4, 4, 4\}$, we can cover $\roundset{p}{S}$ using
    $C'=\{16, 15, 14, 13, 12, 11, 10, 9, 8, 7, 6, 5, 4, 3, 2, 1\}$: items of $\roundset{p}{S}$ that were covered by $(12, 12, 12, 12)$ in $\roundset{p}{C}$ can be covered with $(16, 15, 14, 13)$ in $C'$, whatever was covered by $(8,8,8,8)$ in $\roundset{p}{C}$, can be covered by $(12, 11, 10, 9)$ in $C'$, etc. Hence, if we can cover $\roundset{p}{S}$ with $\roundset{p}{C}$, we can cover $\roundset{p}{S}$ using $C'$; which means we can cover $\roundset{p}{S}$ using $[m']$. 
\end{example}




\newcommand{\lemmaILPFlower}[1]{
\begin{lemma}\label{lemma:exhaustive}
\emph{#1} Given two sets $\roundset{p}{S}$ and $\roundset{p}{C}$, there exists an algorithm that, in time  $p^{\oh(p^3)}\cdot \poly(|S|+|C|)$, either outputs a covering of $\roundset{p}{S}$ using $\roundset{p}{C}$ (if one exists) or correctly reports that no such covering is possible.
\end{lemma}}

\lemmaILPFlower{}

\begin{proof}[sketch]
We model the covering problem as a feasibility instance of integer linear programming
with one variable for each possible way of covering a distinct element in $\roundset{p}{S}$ using one or two elements in $\roundset{p}{C}$. Given that both the ground multiset $\roundset{p}{S}$ and the covering multiset $\roundset{p}{C}$ have constant support $p$, there are only $\oh(p^3)$ such variables.
The constraints enforce that (i) every ground element is covered exactly once,
(ii) no covering type is used more than its multiplicity,
and (iii) only feasible type combinations are allowed.
%
By Lenstra’s algorithm, the resulting integer linear feasibility can be solved in time
$p^{O(p^3)}\cdot \poly(|S|+|C|)$.
Thus, we can either find a feasible covering or correctly conclude that none exists.
\end{proof}

\begin{proof}
Let $V$ and $T$ be the set of distinct elements in $\roundset{p}{S}$ and $\roundset{p}{C}$, respectively.
Then $|V|\le p$ and $|T|\le p$.
For each pair $\{i,j\}$ of items in $T$ and each ground type $v\in V$,
introduce an integer variable $x_{i,j,v}\in \mathbb{Z}_{\ge 0}$,
intended to denote how many times the pair $(i,j)$ is used to cover elements of type $v$. Similarly, for each element $i\in T$ and ground type $v\in V$, we let $x_{i,v} \in \mathbb{Z}_{\ge 0}$ to denote the number of times $i$ is used to cover elements of type $v$.
There are $\big({{|T|}\choose{2}}+|T|\big) |V|\le p^3$ variables.

We write a feasibility ILP whose constraints enforce that
(i) every ground element is covered exactly once, and
(ii) cover-items are not used more than their multiplicities.

\smallskip
\noindent
\emph{Cover all ground elements.}
For each $v\in V$, we let $\mu_S(v)$ be the multiplicity of $v$, that is, $\mu_S(v) := |\{ s\in \roundset{p}{S} : s = v \}|$. To ensure $v$ is covered exactly once, we add a constraint:

\begin{equation}\label{eq:cover-demand}
\sum_{i,j\in T} x_{i,j,v} + \sum_{i\in T} x_{i,v} \;=\; \mu_S(v).
\end{equation}

\smallskip
\noindent
\emph{Do not overuse cover-items.}
For each $t\in T$, let $\mu_C(t)$ be the multiplicity of type $t$ in $\roundset{p}{C}$, that is, $\mu_C(t) := |\{ s\in \roundset{p}{C} : s = t \}|$.
Each time we use a pair $(i,j)$ we consume one item of type $i$ and one item of type $j$
(with the understanding that if $i=j$ then we consume two copies).
Thus, for each $t\in T$ we impose
\begin{equation}\label{eq:cover-supply}
\sum_{v\in V}\ \big( \sum_{i,j\in T}
\bigl(\mathbf{1}[i=t] + \mathbf{1}[j=t]\bigr)\, x_{i,j,v} + \sum_{i\in T}
\bigl(\mathbf{1}[i=t] \big) \, x_{i,v}\bigr)
\;\le\; \mu_C(t).
\end{equation}
Here, $\mathbf{1}[j=t]$ denotes the indicator function, which is equal to $1$ if $j=t$ and $0$ otherwise.

 \smallskip
 \noindent
 \emph{Type-feasibility.}
 Finally, we forbid pairs that cannot cover type $v$ by setting $x_{i,j,v}=0$ whenever $(i,j)$ is not a valid set for covering $v$, that is, $i+j < v$. 
 Similarly, we force $x_{i,v}=0$ when $i<v$. 

\smallskip
\noindent
This ILP has at most $p^3$ integer variables, and the number of constraints is
$\oh(p^3)$.

By Lenstra's theorem~\cite{lenstra1983integer} (and subsequent improvements by Kannan~\cite{kannan1987minkowski}),
integer linear feasibility with a fixed number $d$ of variables can be solved in time $d^{\oh(d)}\cdot \poly(L)$, where $L$ is the encoding length of the ILP.
Here $d\le p^3$ and $L=\poly(|S|+|C|+p)$, hence the feasibility test runs in
$p^{\oh(p^3)}\cdot \poly(|S|+|C|)$ time.
If feasible, the resulting integer solution specifies a valid covering; otherwise, no such solution exists.
\end{proof}

To summarize, given a covering guess size $m$ and ground set $\roundset{p}{S}$, we form $\roundset{p}{C}$, and check whether we can cover $\roundset{p}{S}$ using $\roundset{p}{C}$ with a feasibility ILP (\Cref{lemma:exhaustive}). If we cannot find a covering, we return bad-guess; we cannot cover $\roundset{p}{S}$ using $[m]$. If we find a solution, use the covering given by $\roundset{p}{C}$ to cover $\roundset{p}{S}$ using $[m']$, where $m'=m+\lceil m/p\rceil$. 
This gives a $(1+1/p)$-approximation for covering $\roundset{p}{S}$. Applying the 
\Cref{proposition:optroundset}, we can get an approximating factor of $(1+2/p)(1+1/p)$, which can be as close as we want to $1$ by setting $p$ large enough.
Using binary search adds a multiplicative factor of $\log \opt$ to the runtime we used for ILP (see \Cref{lemma:exhaustive}). Hence, the overall running time is 
$p^{O(p^3)}\cdot \poly(|S|+\opt)$.

We can conclude the following:

\newcommand{\kcycleFinalThm}[1]{
\begin{theorem}\label{thm:twocover}
    \emph{#1} There is a $(1+\epsilon)$-approximation for the \twocover problem that terminates in 
    $p^{O(p^3)}\cdot \poly(|S|+\opt)$,
    where $p=3/\epsilon^2$. 
\end{theorem}
}
\kcycleFinalThm{}
\begin{proof}
    Given an instance $S$ of \twocover with $\epsilon > 0$, set $p = \lceil 3/\epsilon^2 \rceil$ to ensure $(1+2/p)(1+1/p) \leq 1+\epsilon$. Construct $\roundset{p}{S}$ by rounding up elements in $S$. 
    We use binary search over $m$ to find $\opt_p(\roundset{p}{S})$, where $\opt_p(S)$ is the minimum $m$ such that $S$ can be covered using $\roundset{p}{[m]}$. For each candidate set $[m]$:
    \begin{enumerate}
        \item Construct $\roundset{p}{C}$ where $C=[m]$ as per \Cref{def:roundset}.
        \item Use the feasibility ILP
        (\Cref{lemma:exhaustive}) 
        to determine if $\roundset{p}{C}$ can cover $\roundset{p}{S}$ in time $\oh(2^{p^3}|S|(|S|+m))$.
    \end{enumerate}

Therefore, the binary search routine gives the smallest $m$ such that $\roundset{p}{S}$ can be covered using $\roundset{p}{[m]}$.
 %
  %
    We construct a covering of $\roundset{p}{S}$ using $C'=[m']$, where $m' \leq (1+1/p)m = (1+1/p)\opt_p(\roundset{p}{S}) \leq (1+1/p)\opt(\roundset{p}{S})$ by \Cref{claim:roundsetcover}.
    Given that $C'=[m']$ covers $\roundset{p}{S}$,  
    it covers $S$ as well. Using \Cref{proposition:optroundset}, we have $m' \leq (1+1/p)\opt(\roundset{p}{S})\leq (1+1/p)(1+2/p)\opt(S)$.   
    The running time is $\oh(\log \opt)$ binary search iterations, each taking 
    $p^{O(p^3)}\cdot \poly(|S|+\opt)$
    time, giving total complexity 
    $p^{O(p^3)}\cdot \poly(|S|+\opt)$
    where $p^3 = \oh(1/\epsilon^6)$.
\end{proof}

\paragraph*{EPTAS for \telebr in Cycle-Star graphs}

Given the equivalence between \telebr in \kcycle graphs and the \twocover problem, as established in \Cref{lemma:kcycle-equal-twocover}, along with the EPTAS presented in Theorem~\ref{thm:twocover} for the  \twocover problem, we can conclude the following.

\begin{corollary}\label{Corol:PTASkCycle}
There exists an EPTAS for \telebr in \kcycle graphs.  
In particular, for a \kcycle graph on $n$ vertices and any $\epsilon > 0$, the algorithm computes a $(1+\epsilon)$-approximate broadcast schedule in time $O(2^{1/\epsilon^6}\, k(k+n)\log n)$.
\end{corollary}

\noindent \textbf{Extension to multi-source broadcasting} \vspace{1mm} \\
The EPTAS of Corollary~\ref{Corol:PTASkCycle} extends naturally to the multi-source broadcasting setting. First, we provide a sketch of this extension. 

The main idea is to guess the optimal broadcast time $\br$ and, for each guess, attempt to construct a broadcast schedule. If $\br=\bropt(G,S)$, the procedure returns a schedule of length at most $(1+\epsilon)\bropt(G,S)+1$. Trying all relevant guesses for $\br$ and selecting the shortest resulting schedule yields an EPTAS. 

For each fixed guess~$\br$, we reduce the multi-source instance to a single-source instance as follows. We first perform a preprocessing step that removes a window of at most $\br$
vertices around each non-center source.
These vertices can always be informed within $\br+1$ rounds directly from their respective sources.
If this preprocessing disconnects the graph or leaves a vertex at a distance larger than $\br$
from every source, we reject the guess value $\br$.
Otherwise, the remaining graph $G'$ is a graph of distance 1 to a path forest (See Figure~\ref{fig:theirfig}). This graph has the center as the unique bottleneck through which all remaining vertices must be informed.
We therefore reduce the problem to a single-source instance $(G',s)$, where $s$ is the center.
This instance is transformed into a modified version of the \twocover problem in which cycle components allow one or two-element covers as before, while path components (created by preprocessing) correspond to ground elements that must be covered by a single item.
The rounding and approximation steps from the single-source EPTAS remain unchanged, and the resulting rounded instance of \twocover has constant support.
It can therefore be solved exactly using a fixed-dimension ILP as before.  

Now, we present details on the extension of the EPTAS of Corollary~\ref{Corol:PTASkCycle} to the multi-source broadcasting setting. 

Suppose we are given an instance $(G,S)$, where $S$ is the set of all sources, and a fixed parameter $\epsilon>0$. Let $\br$ be a broadcast-time guess value. We describe a procedure that, 
when $\br=\bropt(G, S)$, outputs a valid broadcasting scheme that completes within $(1+\epsilon)\br+1$. Otherwise, it returns either \textsc{Bad-Guess} or a scheme with no guarantee on its broadcast time. 
As we do not know the optimal broadcast time, we iterate over all $\br \in [\lceil \log n \rceil, n] $ and take the scheme with the best broadcasting time among all returned schemes. This way, we can achieve broadcast time $(1+\epsilon)\bropt(G, S) + 1$.

\smallskip
\noindent
\emph{Preprocessing.}
For each non-center source $s_i$, we remove all vertices that lie within distance $\br$ of $s_i$
on either of the two paths connecting $s_i$ to the center, excluding the center itself.
(The vertex $s_i$ is removed as well.)
If the center is not initially a source, let $s_m$ be the source closest to the center,
at distance $d_{\min}$, and remove up to $\br-d_{\min}$ additional vertices along the remaining path between
the center and $s_m$.  \\
Let $G'$ denote the remaining graph.
If $G'$ is disconnected, or if $d_{\min}>\br$, then there exists a vertex whose distance from every source
exceeds $\br$, and we correctly return \textsc{Bad-Guess}. 
Otherwise, $G'$ will be a graph of distance $1$ to path forests (see Figure~\ref{fig:theirfig}).\\  
The preprocessing removes a window of at most $\br$ vertices around each non-center source.
Note that all removed vertices can be informed within $\br+1$ rounds by broadcasting from their respective sources.

\smallskip
\noindent
\emph{Reduction to covering.}
Consider an instance $(G',s)$ of single-source \telebr with the center being the source. 
We associate $(G',s)$ with an instance $\mathcal{S}$ of the \twocover problem,
defined analogously to Lemma~\ref{lemma:kcycle-equal-twocover} for cycle-star graphs.
The only difference is that some elements of the ground multiset $\mathcal{S}$,
which we refer to as \emph{path elements},
correspond to paths in $G'$ rather than cycles.
Each path element must be covered by a subset of size exactly~$1$ (it cannot be covered with two elements from the ground set). \\
This modified covering instance can be solved in the same way as discussed in Section~\ref{subsec:twocover}. In particular, the rounding steps remain unchanged. The only additional constraint arises in the final step
(Lemma~\ref{lemma:exhaustive}),
where we ensure that each path element in the rounded ground set is covered by a single-item subset.
This produces a broadcast schedule for $(G',s)$ that completes within
$(1+\epsilon)\,\bropt(G',s)$ rounds.

\smallskip
\noindent
\emph{Combining the solutions.}
Overall, broadcasting in the original multi-source instance $G$ completes within
\[
\max\!\left\{\, \br+1,\ d_{\min} + 1 +(1+\epsilon)\,\bropt(G',s) \,\right\}.
\]
The term $\br+1$ accounts for informing all vertices removed during the preprocessing step. The second term captures the time required to inform the center and its neighbor along the path toward $s_m$ (namely, $d_{\min}+1$ rounds), followed by the execution of the single-source broadcast within $G'$.

Suppose $br = \bropt(G,S)$. Now, if $G'$ has no vertices other than the center, then we have found a broadcasting schedule with $\br+1 = \bropt(G, S)+1$ rounds. Otherwise, in the optimal scheme, vertices in $G'$ must be informed through the center. 
This is because the vertices in $G'$ are the ones that were not removed when forming $G'$ and thus are at a distance further than $\br=\bropt(G, S)$ from the sources in their respective cycles.
Therefore, as the center cannot be informed earlier than round $d_{\min}$, we can write $\bropt(G, S)\ge d_{\min} +\bropt(G', s)$.
Hence, the broadcast time of our scheme is at most
\[
\max\!\left\{\, \br+1,\ d_{\min} + 1 +(1+\epsilon)\,\bropt(G',s) \,\right\} \leq (1+\epsilon)\bropt(G,S)+1.
\]

\subsection{EPTAS for \Kpath Graphs}
In this section, we present an efficient polynomial-time approximation scheme (EPTAS) for \telebr in \kpath graphs. Consider a \kpath graph $G$ with endpoints $s$ and $t$.
In what follows, we assume that $s$ is the designated source, holding the message at time~0, and that $t$ acts as a \emph{late source}, receiving the message at time~$\alpha$, with $\alpha \leq d(s,t)$, where $d(s,t)$ denotes the distance between $s$ and $t$ in $G$.
We show that any instance of the (single-source) problem with an arbitrary originator can be reduced to this double-source setting (with $s$ as the source and $t$ as the late source) without changing the approximation factor (\Cref{lem:general_case_reduction}).

In Lemma~\ref{lemma:kpath-le-twocovertwo}, we establish an equivalence between double-source \telebr and the integer covering problem that we call \twocovertwo (\Cref{def:twocovertwo}).

Finally, by applying a rounding-and-covering approach similar to the one used for \twocover in \kcycle graphs, we obtain an EPTAS for \twocovertwo, and consequently for \telebr in \kpath graphs (Section~\ref{section:twocovertwo}).
In what follows, we describe these steps in detail. First, we prove the following lemma, which implies we can restrict our attention to the double-source \telebr.

\newcommand{\kpathSpecialCaseLemma}[1]{
\begin{lemma}
\label{lem:general_case_reduction}
\emph{#1} Any $c$-approximation algorithm for \telebr in \kpath graphs in the double-source setting—where $s$ is the primary source and $t$ is a late source receiving the message at time $\alpha$—can be extended to a $c$-approximation algorithm for the (single-source) general setting, where the originator $v_o$ may be any vertex of the graph.
\end{lemma}
}
\kpathSpecialCaseLemma{}

\begin{proof}[sketch]
 Without loss of generality, assume $v_o$ is closer to $s$ than to $t$.
At round $1$, $v_o$ sends the message along its path toward $s$, and at round $2$, it sends the message along the path toward $t$.
Once $s$ receives the message, it immediately forwards it along its shortest path toward $t$.
Finally, when $t$ receives the message, it first transmits it back along its path toward $v$ (to ensure that the neighbor is informed if it has not already been).
These enforced decisions affect the broadcast time only by a constant number of rounds (the approximation remains the same).
If $s$ is informed at time $\tau$, then $t$ is informed at time $\tau' = \tau + \alpha$.
\end{proof}

\begin{proof}
    Let $v_o$ be the originator vertex lying in path $P_i$ at distances $d(v_o,s)$ and $d(v_o,t)$ from endpoints $s$ and $t$, respectively. Without loss of generality, assume $d(v_o,s) \leq d(v_o,t)$. 
    The message is first forwarded along the shortest path from $v_o$ to $s$, which takes $d(v_o,s)$ rounds. At round $d(v_o,s)+1$, vertex $s$ is informed and can begin broadcasting. Vertex $t$ is informed either directly through $v_o$ (after $d(v_o,t)$ rounds) or through $s$ via some other path $P_j$ (whichever occurs first). Let $\alpha$ denote the delay with which $t$ receives the message relative to $s$. Once both $s$ and $t$ are informed, they proceed with the $c$-approximation algorithm on the remaining graph.

    We consider two cases depending on how $t$ receives the message:

    \textbf{Case I: $t$ is informed directly through $v_o$.} 
    Here $\alpha = d(v_o,t) - d(v_o,s) + 1$.
    If the global broadcast time is determined by a vertex in $P_i$, then the scheme is optimal. Otherwise, the process completes within $d(v_o,s) + c \cdot \br^-(\alpha)$ rounds, where $\br^-(\alpha)$ is the optimal broadcast time for the subproblem excluding $P_i$ with sources $s$ and $t$ where $t$ has delay $\alpha$. We have $\br \geq d(v_o,s) + \br^-(\alpha-1)\geq d(v_o,s) -1+ \br^-(\alpha)$ because $v_o$ can inform its neighbor toward $t$ in its first round, which results in $t$ getting informed one round sooner.
    Therefore, this yields the same approximation factor $c$.

    \textbf{Case II: $t$ is informed through $s$ via another path $P_j$.}
    Let $\ell_j$ be the number of vertices of path $P_j$ (excluding $s$ and $t$). Then $t$ is informed at round $d(v_o,s) + 1 + \ell_j$, giving delay $\alpha = \ell_j$.
    If the global broadcast time is determined by a vertex in $P_i$ or $P_j$, then the scheme is optimal. Otherwise, the process completes within $d(v_o,s) + 1 + c \cdot \br^-(\alpha)$ rounds, where $\br^-(\alpha)$ is the optimal broadcast time for the subproblem excluding $P_i$ and $P_j$ with sources $s$ and $t$ where $t$ has delay $\alpha$. Since $\br \geq d(v_o,s) + 1 + \br^-(\alpha)$, this yields the same approximation factor $c$.
\end{proof}

\paragraph*{From \telebr in Melon Graphs to \twocovertwo}
We introduce the following problem, which is the key part of our EPTAS. 

\begin{definition}\label{def:twocovertwo}
The input to the \twocovertwo problem is a \emph{ground multiset} $S=\{\ell_1,\ell_2,\ldots,\ell_n\}$ of positive integers together with a non-negative integer $\beta$. 
The objective is to determine the minimum integer $m$ such that there exist functions  
$c : S \to \{0\}\cup[m]$ and $d : S \to \{0\}\cup[m-\beta]$  
satisfying $c(i) + d(i) \geq \ell_i$ for all $i \in [n]$.  
Moreover, injectivity is required in the following sense: for any $i \neq j$, if both $c(i)> 0$ (respectively, $d(i)>0$), then $c(i)\neq c(j)$ (respectively $d(i)\neq d(j)$).  
In this setting, we say that each $\ell_i$ is \emph{covered} by the pair $\{c(i),d(i)\}$, and we refer to $C=[m]$ and $D=[m-\beta]$ as the \emph{covering sets} of $S$. 
\end{definition}

\begin{example}
    When $S=\{8,4,4,2\}$ and $\beta = 2$, the optimal solution is given by $C=[5]$, and we have $c(1)=5, d(1)=3$; $c(2)=4, d(2)= 0$, $c(3)=2, d(3)=2$,  $c(4) = 1, d(4)=1$. 
\end{example}

In what follows, we establish the equivalence between double-source \telebr in \kpath graphs and   \twocovertwo (\Cref{def:twocovertwo}). 


\begin{lemma}\label{lemma:kpath-le-twocovertwo}
Let $G$ be a \kpath graph formed by $k$ paths with common endpoints $s$ and $t$, where the $i$th path contains $\ell_i \geq 0$ internal vertices (in addition to $s$ and $t$).  
Consider a double-source \telebr instance on $G$, where $s$ receives the message at time~0 and $t$ at time~$\alpha$.  
Now, define an instance $I$ of \twocovertwo as
$I = (S=\{\ell_1,\ldots,\ell_k\},\; \beta = \alpha - 1)$.
Then computing the minimum broadcast time for this double-source \telebr instance is equivalent to solving the \twocovertwo instance $I$.  
\end{lemma}

\begin{proof}
Consider the covering of the instance $I$ with parameter $m$. 
We claim that there exists a valid broadcasting scheme for $G$ within $\br = m$ rounds. 
For each $i \in [k]$, let $c(i)$ and $d(i)$ denote the assigned values in the covering.  
Then $s$ informs the $i$th path at round $m - c(i) + 1$ and $t$ at round $m - d(i) + 1$.  
Thus, $s$ (respectively $t$) uses disjoint rounds drawn from 
$\{\, m - c + 1 \mid c \in C \,\} = [m]$  (respectively from 
$\{\, m - d + 1 \mid d \in [m - \alpha + 1] \,\} = [\alpha, m]$).  
Hence, both $s$ and $t$ follow a valid broadcast scheme.  

It remains to argue coverage.  
Since $c(i) + d(i) \geq \ell_i$, the remaining time after the calls at rounds $m - c(i) + 1$ and $m - d(i) + 1$ is sufficient to cover all $\ell_i$ internal vertices of path $i$.  
Thus, the covering implies a feasible broadcasting scheme.

Conversely, suppose we have a broadcasting scheme using $\br$ rounds.  
For each $\ell_i \in S$, if $s$ and $t$ call vertices in path $i$ at rounds $y^i_s$ and $y^i_t$, define 
$c(i) = \br - y^i_s + 1$ and $d(i) = \br - y^i_t + 1$.  
Since calls at different paths use distinct rounds, the values $c(i)$ (resp.\ $d(i)$) are all distinct.  
Moreover, the chosen values lie in  
$\{\br - x + 1 \mid x \in [\br]\} = [\br]$  
and  
$\{\br - x + 1 \mid x \in [\alpha,\br]\} = [\br - \beta]$, respectively.  
Finally, feasibility of the broadcasting scheme ensures  
$c(i) + d(i) = (\br - y^i_s + 1) + (\br - y^i_t + 1) \geq \ell_i$,  
so the functions $c$ and $d$ indeed form a valid covering of $S$.  
This establishes the equivalence.
\end{proof}

\paragraph*{EPTAS for \twocovertwo} \label{section:twocovertwo}
In the light of Lemmas~\ref{lem:general_case_reduction} and \ref{lemma:kpath-le-twocovertwo}, 
presenting an EPTAS for \telebr in \kpath graphs reduces to presenting an EPTAS for the \twocovertwo problem. 
This is the focus of the present section.  
Specifically, we prove \Cref{thm:twocovertwo}, which establishes the existence of a $(1+\epsilon)$-approximation algorithm 
for \twocovertwo for any $\epsilon > 0$. 
Our approach closely parallels the rounding and covering techniques used in Section~\ref{subsec:twocover} for \kcycle graphs.  In particular, it also addresses the covering problem through a two-stage approximation scheme that leverages the concept of rounded sets (see \Cref{def:roundset}) to reduce the problem complexity. 
The main idea is to transform the original ground multiset and the two covering sets into simplified versions with constant support (i.e., bounded numbers of distinct elements), so that an ILP formulation becomes computationally feasible while still preserving provable approximation guarantees.

Our EPTAS for \twocovertwo proceeds by reducing the original instance to one with bounded support.  
Given an input multiset $S$, we first form its \emph{rounded version} $\roundset{p}{S}$, where each element is rounded up to a larger element of $S$ so that the number of distinct values becomes $p$.  
Here $p$ depends only on $\epsilon$.  
Although rounding inflates item sizes, we show that the optimal covering cost for $\roundset{p}{S}$ is at most a $(1+1/p)$ factor larger than that of $S$ (\Cref{proposition:optroundset:2}). 
Thus, solving the rounded instance suffices for approximation.  

To search for the optimal covering, we use binary search over candidate values $m$, where $m$ determines the covering sets $C=[m]$ and $D=[m-\beta]$.  
For each candidate $m$, we construct rounded sets $\roundset{p}{C}$ and $\roundset{p}{D}$, each with support size $p$.  
We prove that if $\roundset{p}{S}$ can be covered using $\roundset{p}{C}$ and $\roundset{p}{D}$, then the original instance can be covered using $C'=[m']$ and $D'=[m'-\beta]$ with $m' \leq (1+1/p)m$ (\Cref{claim:roundsetcover:2}).  
Conversely, if we cannot cover $\roundset{p}{S}$ with the rounded sets, then it is not possible to cover it using $[m]$ and $[m-\beta]$.  
It therefore suffices to test whether $\roundset{p}{S}$ is coverable using $\roundset{p}{C}$ and $\roundset{p}{D}$.  
Since all three sets have constant support (bounded by $p$), this check can be carried out by an ILP formulation in polynomial time (\Cref{lemma:exhaustive2}).  
The remainder of this section provides details of the algorithm.

\newcommand{\kpathpropoptroundset}[1]{
\begin{proposition}\label{proposition:optroundset:2}
    \emph{#1} $\opt(\roundset{p}{S}) \leq \opt(S) + \lceil |S|/p\rceil$, where $\opt$ is the optimal cover for a specific instance of \twocovertwo. 
\end{proposition}
}
\kpathpropoptroundset{}
\begin{proof}
    By definition, $\roundset{p}{S}$ consists of $p$ parts.
    Given an optimal cover of $S$, form a cover of $\roundset{p}{S}$ by assigning the numbers given to the part $i$ in $S$ to the part $(i+1)$ in $\roundset{p}{S}$ for $i\in[p-1]$. As the size and value of part $i+1$ is not greater than the value of part $i$, the set of numbers that can cover part $i$ can also cover part $i+1$. 

    Therefore, the first part is the only part that needs to be covered. We show that adding an extra $\lceil |S|/p \rceil$ to the covering set is enough to cover the first part. This is because every $\ell_i\in S$ is less than 
    $2\opt(S)-\beta$, 
    as otherwise it would not be possible to cover $v_i$ with a number in $C$ and a number in $D$ for the optimal value of $m$. Thus, by adding an extra $ \lceil |S|/p \rceil$ items to the $m$, we can cover $\roundset{p}{S}$.
\end{proof}

Since $\opt(S) \geq |S|$, it follows that any covering of $\roundset{p}{S}$ yields a $(1+1/p)$-approximation for covering $S$.  
This allows us to restrict attention to $\roundset{p}{S}$, which has constant support $p$.  
To cover $\roundset{p}{S}$, we apply a binary search strategy.  
For a candidate value $m$, we design a procedure that either produces a covering using $[m + \lceil m/p \rceil]$, or certifies that no covering exists with $[m]$. 
This way, we identify the smallest value of $m$ for which $\roundset{p}{S}$ can be covered using the rounded covering sets $\roundset{p}{C}$ and $\roundset{p}{D}$, where $C=[m]$ and $D=[m-\beta]$.


 

\newcommand{\kpathclaimroundsetcover}[1]{
\begin{claim}\label{claim:roundsetcover:2}
    {#1} Given that we can cover $S$ with $\roundset{p}{C}$ and $\roundset{p}{D}$, we can cover $S$ using $C'=[m']$ and $D'=[m'-\beta]$, where $m'=m+\lceil m/p \rceil$.
\end{claim}
}
\kpathclaimroundsetcover{}
\begin{proof}
    Based on definition, $\roundset{p}{C}$ contains $p$ parts. 
    Suppose that we have a covering with $\roundset{p}{C}$ and $\roundset{p}{D}$ for $S$.
    Associate value $m+\lceil m/p \rceil - (i-1)$ to the $i$'th largest number in $C$ for $i\in[m]$. 
    This transformation only increases the value as the first block with original value $m$ matches to $m+1, \ldots, m+\lceil m/p \rceil$, the second block matches to $m-\lceil m/p \rceil +1 ,\ldots, m$, and so on. More specifically, the $i$'th block matches to $m-(i-1)\lceil m/p \rceil +1 ,\ldots, m-(i-2)\lceil m/p \rceil$, which is greater than the original value of the block. 
    Similarly, associate value $m+\lceil m/p \rceil -\beta - (i-1)$ to the $i$'th largest number in $D$ for $i\in[m-\beta]$. 
    As the new sets are subsets of $C'=[m']$ and $D'=[m'-\beta]$, and every value in $\roundset{p}{C}$ and $\roundset{p}{D}$ is now larger than its original value, the same covering can cover $S$ with the corresponding values in $C'$ and $D'$. 
\end{proof}

\newcommand{\kpathlemmaexhaustive}[1]{

\begin{lemma}\label{lemma:exhaustive2} \emph{#1} 
Given three sets $\roundset{p}{S}$, $\roundset{p}{C}$, and $\roundset{p}{D}$, there exists an algorithm that, in time $p^{\oh(p^3)}\cdot \poly(|S|+|C|+|D|)$, either outputs a covering of $\roundset{p}{S}$ using the covering sets $\roundset{p}{C}$ and $\roundset{p}{D}$ or correctly reports that no such covering is possible.
\end{lemma}

}
\kpathlemmaexhaustive{}

\begin{proof}
Let $V$, $T_C$, and $T_D$ be the set of distinct elements in $\roundset{p}{S}$, $\roundset{p}{C}$, and $\roundset{p}{D}$, respectively.
Then $|V|\le p$ and $|T_C|, |T_D|\le p$.
For each tuple $(i,j)\in (T_C\cup \{\emptyset\}) \times (T_D \cup \{\emptyset\})$ (except $(\emptyset, \emptyset)$) and each ground type $v\in V$,
introduce an integer variable $x_{i,j,v}\in \mathbb{Z}_{\ge 0}$,
intended to denote how many times the pair $(i,j)$ is used to cover elements of type $v$. 
There are $((|T_C|+1)(|T_D|+1)-1) |V|\le p^3$ variables.

We write a feasibility ILP whose constraints enforce that
(i) every ground element is covered exactly once, and
(ii) cover-items are not used more than their multiplicities,
(iii) every ground element is covered with the right set of values.

\smallskip
\noindent
\emph{Cover all ground elements.}
For each $v\in V$, we let $\mu_S(v)$ be the multiplicity of $v$, that is, $\mu_S(v) := |\{ s\in \roundset{p}{S} : s = v \}|$. To ensure $v$ is covered exactly once, we add a constraint:

\begin{equation}\label{eq:cover-demand}
\sum_{(i,j)\in (T_C\cup \{\emptyset\}) \times (T_D \cup \{\emptyset\})} x_{i,j,v}\;=\; \mu_S(v).
\end{equation}

\smallskip
\noindent
\emph{Do not overuse cover-items.}
For each $t\in T_C$, let $\mu_C(t)$ be the multiplicity of type $t$ in $\roundset{p}{C}$, that is, $\mu_C(t) := |\{ s\in \roundset{p}{C} : s = t \}|$. Similarly, define $\mu_D(t)$ for each $t\in T_D$.
Each time we use a pair $(i,j)$ we consume one item of type $i$ and one item of type $j$
(with the understanding that if $i=j$ then we consume two copies).
Thus, for each $t\in T_C$ we impose
\begin{equation}\label{eq:cover-supply}
\sum_{v\in V}\ \sum_{j\in T_D\cup \{\emptyset\}}
\, x_{t,j,v}
\;\le\; \mu_C(t).
\end{equation}
Similarly, we impose the same constraint for $\mu_D(t)$ for each $t \in T_D$:
\begin{equation}\label{eq:cover-supply}
\sum_{v\in V}\ \sum_{i\in T_C\cup \{\emptyset\}}
\, x_{i,t,v}
\;\le\; \mu_D(t).
\end{equation}

 \smallskip
 \noindent
 \emph{Type-feasibility.}
 Finally, we forbid pairs that cannot cover type $v$ by setting $x_{i,j,v}=0$ whenever $(i,j)$ is not a valid set for covering $v$, that is, when the sum of $i$ and $j$ (if it exists) is less than $v$.  

\smallskip
\noindent
This ILP has at most $p^3$ integer variables, and the number of constraints is
$\oh(p^3)$.

By Lenstra's theorem~\cite{lenstra1983integer} (and subsequent improvements by Kannan~\cite{kannan1987minkowski}),
integer linear feasibility with a fixed number $d$ of variables can be solved in time $d^{\oh(d)}\cdot \poly(L)$, where $L$ is the encoding length of the ILP.
Here $d\le p^3$ and $L=\poly(|S|+|C|+|D|+p)$, hence the feasibility test runs in
$p^{\oh(p^3)}\cdot \poly(|S|+|C|+|D|)$ time.
If feasible, the resulting integer solution specifies a valid covering; otherwise, none exists.
\end{proof}

To summarize, for a candidate covering size $m$ and ground set $\roundset{p}{S}$, we construct $\roundset{p}{C}$ and $\roundset{p}{D}$ and test, via a feasibility ILP (\Cref{lemma:exhaustive2}), whether $\roundset{p}{S}$ can be covered.  
If no covering is found, we conclude that $\roundset{p}{S}$ cannot be covered with $[m]$ (bad guess).  
If a covering is found, we extend it to a covering with $m' = m + \lceil m/p \rceil$, thereby obtaining a $(1+1/p)$-approximation for $\roundset{p}{S}$.  
Combining this with \Cref{proposition:optroundset:2}, we achieve an overall approximation factor of $(1+1/p)^2$, which approaches~1 as $p$ grows.  


\newcommand{\kpathfinalthm}[1]{
\begin{theorem}\label{thm:twocovertwo}
    \emph{#1} There is a $(1+\epsilon)$-approximation for the \twocovertwo problem that terminates in $\oh(p^{\oh(p^3)}poly(|S|,\opt))$, where $p=3/\epsilon^2$. 
\end{theorem}
}
\kpathfinalthm{}
\begin{proof}
    Given an instance $S$ of \twocovertwo with $\epsilon > 0$, set $p = \lceil 3/\epsilon^2 \rceil$ to ensure $(1+1/p)^2 \leq 1+\epsilon$. Construct $\roundset{p}{S}$ by rounding up elements in $S$. 
    We use binary search over $m$ to find $\opt_p(\roundset{p}{S})$, where $\opt_p(S)$ is the minimum $m$ such that $S$ can be covered using $\roundset{p}{[m]}$ and $\roundset{p}{[m-\beta]}$. For each candidate set $[m]$:
    \begin{enumerate}
        \item Construct $\roundset{p}{C}$ and $\roundset{p}{D}$ where $C=[m], D=[m-\beta]$ as per \Cref{def:roundset}.
        \item Use the feasibility ILP (\Cref{lemma:exhaustive2}) to determine if $\roundset{p}{C}$ and $\roundset{p}{D}$ can cover $\roundset{p}{S}$ in time $\oh(p^{O(p^3)}poly(|S|,\bropt)$.
    \end{enumerate}

    Therefore, the binary search routine gives the smallest $m$ such that $\roundset{p}{S}$ can be covered using $\roundset{p}{[m]}$ and $\roundset{p}{[m-\beta]}$.
%
 %
    We construct a covering of $\roundset{p}{S}$ using $C'=[m']$ and $D'=[m'-\beta]$, where $m' \leq (1+1/p)m = (1+1/p)\opt_p(\roundset{p}{S}) \leq (1+1/p)\opt(\roundset{p}{S})$ by \Cref{claim:roundsetcover:2}.
    Given that $C'=[m']$ and $D'=[m'-\beta]$ cover $\roundset{p}{S}$,  
    they cover $S$ as well. Using \Cref{proposition:optroundset:2}, we have $m' \leq (1+1/p)\opt(\roundset{p}{S})\leq (1+1/p)^2\opt(S)$.   
    The running time is $\oh(\log \opt)$ binary search iterations, each taking $\oh(p^{\oh(p^3)}poly(|S|,\bropt))$ time, giving total complexity $p^{\oh(p^3)}poly(|S|,\bropt))$ where $p^3 = \oh(1/\epsilon^6)$.
\end{proof}
\paragraph*{EPTAS for \telebr in Melon graphs}

Given the equivalence between \telebr in \kpath graphs and the \twocovertwo problem, as established in \Cref{lemma:kpath-le-twocovertwo} along with the EPTAS presented in Theorem~\ref{thm:twocovertwo} for the \twocovertwo problem, we can conclude the following.


\begin{corollary}\label{Coro:PTASKPath}
There exists an EPTAS for \telebr in \kpath graphs.  
In particular, for a \kpath graph on $n$ vertices and any $\epsilon > 0$, the algorithm computes a $(1+\epsilon)$-approximate broadcast schedule in time $O(2^{1/\epsilon^6}\, k(k+n)\log n)$.
\end{corollary}

\subsubsection*{Multi-source Broadcasting in \Kpath Graphs}\label{apx:multisource-melon}

The EPTAS of  Corollary~\ref{Coro:PTASKPath} extends naturally to the multi-source broadcasting setting. Suppose we are given an instance $I=(G, S)$ of multi-source broadcasting in $G$, where $S$ is the set of all sources, along with a fixed parameter $\epsilon>0$. Let $\br$ be a broadcast-time guess value. We describe a procedure that, 
when $\br=\bropt(G, S)$, outputs a valid broadcasting scheme that completes within $(1+\epsilon)\br+1$. Otherwise, it returns either \textsc{Bad-Guess} or a scheme with no guarantee on its broadcast time. 
As we do not know the optimal broadcast time, we iterate over all $\br \in [\lceil \log n \rceil, n] $ and take the scheme with the best broadcasting time among all returned schemes. This way, we can achieve broadcast time $(1+\epsilon)\bropt(I) + 1$.

\smallskip
\noindent
\emph{Preprocessing.}
For each non-endpoint source $s_i$ in the \kpath graph, we remove all vertices that lie within distance $\br$ of $s_i$
on their respective paths, excluding the $s$ and $t$.
(The vertex $s_i$ is also removed.)
If the endpoints are not initially source vertices, let $s_{ms}$ and $s_{mt}$ be the sources closest to $s$ and $t$, respectively,
at distance $d_{\min}(s)$ and $d_{\min}(t)$.  \\
Let $G'$ denote the remaining graph and $d_{\min} = \min(d_{\min}(s), d_{\min}(t))$.
If $G'$ is disconnected, or if $\max(d_{\min}(s),d_{\min}(t))>\br$, then there exists a vertex whose distance from every source
exceeds $\br$, and we correctly return \textsc{Bad-Guess}. \\
%
The preprocessing removes a window of at most $\br$ vertices around each non-center source.
Note that all removed vertices can be informed within $\br+1$ rounds by broadcasting from their respective sources. 

\smallskip
\noindent
\emph{Reduction to covering.}
Consider an instance $I'=(G',\{s,t\}, \alpha)$ of \telebr, where $s$ receives the message at time $0$ and $t$ at time $\alpha$.
We associate $I'$ with an instance $\mathcal{S}$ of the \twocovertwo problem,
defined analogously to Lemma~\ref{lemma:kpath-le-twocovertwo} for \kpath graphs.
The only difference is that some elements of the ground multiset $\mathcal{S}$,
which we refer to as \emph{partial path elements},
correspond to partially-covered paths in $G'$ rather than full paths.
Each such path, depending on the endpoint that it is connected to, must be covered by exactly one of the endpoints (it cannot be covered by both covering sets). \\
This modified covering instance can be solved in the same way as discussed in Section~\ref{section:twocovertwo}. In particular, the rounding steps remain unchanged. The only additional constraint arises in the final step
(Lemma~\ref{lemma:exhaustive2}),
where we ensure that each partial path element in the rounded ground set is covered by its corresponding rounded cover set.
This produces a broadcast schedule for $I'$ that completes within
$(1+\epsilon)\,\bropt(I')$ rounds.

\smallskip
\noindent
\emph{Combining the solutions.}
Overall, broadcasting in the original multi-source instance $G$ completes within
\[
\max\!\left\{\, \br+1,\ d_{\min} + 1 +(1+\epsilon)\,\bropt(I') \,\right\}.
\]
The term $\br+1$ accounts for informing all vertices removed during the preprocessing step. The second term captures the time required to inform the closest endpoint and possibly $1$ extra round if the same source covers both endpoints (one of the endpoints can be informed through the other endpoint), which is $d_{\min}+1$ rounds, followed by the execution of the single-source broadcast within $I'$.

Suppose $br = \bropt(I)$. Now, if $G'$ has no vertices other than $s$ and $t$, then we have found a broadcasting schedule with $\br+1 = \bropt(I)+1$ rounds. Otherwise, in the optimal scheme, vertices in $G'$ must be informed through the $s$ and $t$. 
This is because the vertices in $G'$ are the ones that were not removed when forming $G'$ and thus are at a distance further than $\br=\bropt(I)$ from the sources in their respective cycles.
Therefore, as the center cannot be informed earlier than round $d_{\min}$, we can write $\bropt(I)\ge d_{\min} +\bropt(I')$.
Hence, the broadcast time of our scheme is at most
\[
\max\!\left\{\, \br+1,\ d_{\min} + 1 +(1+\epsilon)\,\bropt(G',s) \,\right\} \leq (1+\epsilon)\bropt(G,S)+1.
\]


\section{Graphs of Bounded Cutwidth}

In this section, we give a polynomial-time algorithm for optimally solving the multi-source broadcasting problem on graphs of bounded cutwidth. Intuitively, the cutwidth constraint limits the degree of $G$ (in contrast to other families such as \kcycle or \kpath graphs), and a cutwidth ordering exposes a small ``active window'' of vertices at each step. This property enables a dynamic programming (DP) approach.
We assume throughout that a cutwidth-$k$ ordering $(v_1,\ldots,v_n)$ is given; all references to ``before'' and ``after'' are with respect to this ordering. Recall that such an ordering can be computed in fixed-parameter tractable time for constant~$k$~\cite{thilikos2005cutwidth}.



\medskip
\noindent\textbf{DP Table.}
We define a DP table with entries of the form $dp[i][s]$, where $i \in [0,n]$ and $s$ is a \emph{state}.  
The value $dp[i][s]$ denotes the minimum broadcast time for informing the first $i$ vertices in the cutwidth ordering while ending in state $s$.  
Define the set of \emph{boundary vertices} as $B_i = \{v_j| j\leq i, v_j \text{ has at least a neighbor after }v_i\}$. Given that cutwidth is $k$, it holds that $|B_i|\leq k$.
The state $s$ encodes all necessary information about the {boundary vertices}, namely, how these vertices are informed and how the vertices appearing after $v_i$ can connect to them to form a potential broadcast tree.

Intuitively, a state describes the partial structure of a broadcast forest spanning boundary vertices together with their timing information. 
We fill out the table row by row, starting with $i=0$ and ending with $i=n$.  
The last row corresponds to the entire graph, and the minimum value among its states indicates the optimal broadcast time.

\medskip
\noindent\textbf{Definition of states.}
Fix $i$, and recall $B_i$ denotes the set of boundary vertices. We define the \emph{active window} of $i$, denoted by $A_i$, as a set formed by the union of $B_i$ and all neighbors of vertices in $B_i$.
%
Given a cutwidth-$k$ ordering, the degree of every vertex is at most $2k$.
For an index $i$, the active window contains all vertices in $B_i$ (at most $k$ vertices) and their neighbors.
Therefore, its total size is bounded by $k + 2k^2$.

A state $s$ encodes a potential partial broadcast tree consistent with the active window by storing the following information:
\begin{itemize}
    \item \textbf{Parents.} For each $u\in B_i$, its parent $p(u)$ (unless it is a source vertex).  
    Each $p(u)$ is chosen from at most $2k$ active vertices or a special ``no parent'' symbol for the sources, yielding $(2k+1)^k$ possibilities.
    \item \textbf{Ranks.} For each $u\in B_i$, its rank among the children of $p(u)$, ordered by the rounds in which $p(u)$ calls them.  
    Since $p(u)$ has at most $2k$ children, there are at most $2k$ options per vertex, and $(2k)^k$ total options.
    \item \textbf{Outgoing calls.} For each $u\in B_i$, the active vertex that it calls in each round after being informed.  
    For every round slot, there are $2k+1$ options (one of up to $2k$ children or a null option if the slot is unused), giving at most $(2k+1)^{2k+1}$ possibilities per $u\in B_i$.
    \item \textbf{Inform times.} The time at which the root of each subtree is informed.  
    Since there are at most $k$ subtrees (each containing a boundary vertex), and each root may be informed at a round in $[n]$, this yields $n^k$ possibilities.
\end{itemize}

\begin{proposition}\label{prop:stateNumber}
If $S$ denotes the space of all states, then
\begin{align*}
|S| & \;\le\; f(k) \cdot n^k,\\ \text{\ where \ \ }  f(k) &= (2k+1)^k \cdot (2k)^k \cdot (2k+1)^{k(2k+1)}.
\end{align*}
In particular,  $|S| = \oh(n^k)$ for constant $k$.
\end{proposition}

For a state $s$, let $\ell(s)$ denote the latest time at which any active node in the broadcast subtrees of $s$ is informed.  
Note that the global broadcast time of a scheme realizing $s$ may exceed $\ell(s)$ if the last informed vertex lies outside the active window.

\medskip
\noindent\textbf{Filling Out the DP table.}
A state $s$ is \emph{valid} for step $i$ if its parent–child assignments among nodes in its broadcast forest respect the edges of $G$.  
For example, if $u_1$ and $u_2$ are not adjacent in $G$ but $u_1$ is the parent of $u_2$ in $s$, then $s$ is invalid at~$i$. Furthermore, the inform times of all source vertices must be $0$.
For any invalid state at step $i$, we set $dp[i][s]=+\infty$.  

\smallskip
\noindent\emph{Base case.}  
When $i=0$, the active window is an empty set. Therefore, there is only one \emph{null} valid state $s_\phi$, for which $dp[0][s_\phi] =0$.  

\smallskip
\noindent\emph{Transitions.}  
To compute $dp[i+1][s]$ for a valid state $s$, we check states like $s'$ that are valid at step $i$ and \emph{compatible} with $s$.  
Compatibility means that information concerning the overlap between active states at steps $i$ and $i+1$ matches in both states:  the inform times of shared vertices, their parent–child relations, and the directions and timings of edges incident to them must all agree.  

If $s^*$ is a compatible state with $s$ with minimum $dp[i][s^*]$, we set
$
dp[i+1][s] \;=\; \max\{ dp[i][s^*], \ell(s)\}.
$
We also record $\mathrm{previous}(s)=s^*$, which allows reconstruction of the optimal broadcast tree.  
If no compatible state exists, then $dp[i+1][s]=+\infty$.  
Intuitively, realizing state $s$ at $i+1$ requires extending from some compatible predecessor state at $i$, and the broadcast time is determined by whichever is larger: the time already incurred or the latest inform time among the active vertices.




\medskip

Using an inductive argument, 
one can establish the correctness of the above approach.  

\newcommand{\bandwidthMainTheorem}[1]{
\begin{theorem}\label{Th:BandWidthMain}\emph{#1}
\telebr can be solved optimally on graphs of cutwidth $k$ in time $f(k)^2\cdot n^{2k+1}$.  
In particular, for a constant $k$, there is a polynomial-time algorithm for computing an optimal broadcast schedule.
\end{theorem}
}

\bandwidthMainTheorem{}
\begin{proof}

\textbf{Setup.}
Recall that we have a fixed cutwidth-$k$ ordering $(v_1,\dots,v_n)$. 
As discussed, at step $i$, the DP table has states $s$ that summarize the broadcast restricted to the \emph{active window} $A_i$ with boundary set $B_i$. 
Recall that, for a state $s$, $\ell(s)$ is the latest inform time among active vertices encoded by $s$.

\medskip
\noindent\textbf{Correctness.}
Define the \emph{overlap signature} of a state $s$ at step $i$ as the restriction
$\pi_i(s)$ of $s$ to the overlap $A_i\cap A_{i+1}$: this includes (i) parent/child relations among vertices in the overlap, (ii) the directions/times of incident used edges, and (iii) their inform times. 
By definition, two states $s'$ at $i$ and $s$ at $i{+}1$ are \emph{compatible} iff $\pi_i(s')=\pi_i(s)$. We use induction to prove the following invariant:

\smallskip\noindent
\emph{Claim 1 (Inductive invariant).}
For every $i\ge 0$ and every valid state $s$ at step $i$, the entry $dp[i][s]$ of the DP table is the minimum broadcast time of any feasible broadcast scheme for informing $v_1,\ldots, v_i$ which ends up at state $s$ for the active set.

\emph{Base $i=0$.}
As there are no vertices in this case, the answer for state $s_\phi$ is trivially $0$ to inform an empty set. 

\emph{Step $i\to i+1$.}
Let $s$ be any valid state at step $i+1$.
%
Choose a compatible predecessor $s^*$ at step $i$ that minimizes $dp[i][s^*]$. By the inductive hypothesis there is broadcast scheme $\mathcal{S}^*$ that realizes $s^*$ and informs $v_1,\ldots,v_i$  in $dp[i][s^*]$ rounds. Extend $\mathcal{S}^*$ inside the active window to realize $s$ at step $i{+}1$ (this is possible by compatibility of $s$ and $s^*$ and their validity). The broadcast time of the extended scheme is $\max\{dp[i][s^*],\ell(s)\}$, which is exactly the value assigned by the recurrence. Hence, $dp[i{+}1][s]$ upper-bounds the true optimum among schedules realizing $s$.

On the other hand, let $\mathcal{S}$ be any feasible schedule on $G_{i+1}$ that realizes $s$ with minimum broadcast time $\tau$. 
Let $s'$ be the state of $\mathcal{S}$ at step $i$; then $s'$ is valid at $i$ and compatible with $s$. By the inductive hypothesis, $dp[i][s']$ is no more than the broadcast time of $\mathcal{S}$ for informing $v_1,\ldots,v_i$. 
Since realizing $s$ from $s'$ cannot finish later than $\max\{dp[i][s'],\ell(s)\}$, we get
$
dp[i{+}1][s]\le \max\{dp[i][s'],\ell(s)\}\le \tau.$
Thus $dp[i{+}1][s]$ lower-bounds the optimum as well. This proves Claim 1.
We conclude that taking the minimum $dp[n][s]$ over valid states $s$ yields the optimal broadcast time on $G$.


\medskip
\noindent\textbf{Running time.}
Let $S$ be the set of all states. 
By Proposition~\ref{prop:stateNumber}, we have $
|S|\le (2k{+}1)^k\cdot(2k)^k\cdot(2k+1)^{k(2k+1)}\cdot n^k \;=\; f(k)\,n^{k}.$
That is, the size of the DP table would be $|S|n = f(k) n^{k+1}$.
To set each index of the DP table, we iterate over all possible states, which takes $\oh(|S|)$ time. Overall, this yields a total time of $\oh(f(k)^2\, n^{2k+1})$.
\end{proof}

Note that the above dynamic programming approach is formulated for instances with multiple sources, and therefore Theorem~\ref{Th:BandWidthMain} extends to the multi-source setting.

\section{Improved Approximation Algorithm for Split Graphs}
In this section, we present a polynomial-time algorithm that finds a broadcasting scheme which completes within $\apxsplit  \bropt$ rounds for split graphs.

Consider an instance $(G,s)$ of \telebr in an input split graph $G=(V,E)$ whose vertex set can be partitioned into a clique $C$ and an independent set $I$. 
Throughout, we use $B$ 
to denote the bipartite graph induced by the edges incident to $I$. \\ We begin by introducing the necessary preliminaries and then describe our algorithm. \vspace{2mm}

\noindent\textbf{Minimum-degree cover.}
An important concept related to broadcasting on split graphs is the \emph{minimum-degree cover}, which also plays a key role in our algorithm.
A \emph{cover} of $B$ is a subset $M$ of edges of $B$ such that every vertex $u\in I$ is incident to exactly one edge of $M$.
For a clique vertex $v\in C$, let 
$\deg_M(v)$ denote the degree of $v$ under $M$, and define the {degree} of $M$ as
$\Delta(M) := \max_{v\in C} \deg_M(v)$.
Let $d^*(B) := \min\{ \Delta(M) \mid M \text{ is a cover of } B \}$.
We refer to a cover attaining degree $d^*(B)$ as a
\emph{minimum-degree cover} of $B$. When the bipartite graph $B$ is clear from the context, we write $d^*$ instead of $d^*(B)$ (See Figure~\ref{fig:covers}).

The value $d^*$ can be computed in polynomial time using a straightforward max-flow formulation.
Introduce a source $s'$, a sink $t'$, connect $s'$ to each $u\in C$ with capacity~$D$,
connect each $(u,v)\in E(B)$ with capacity~$1$,
and connect each $v\in I$ to $t'$ with capacity~$1$.
The smallest value $D$ for which a flow of value $|I|$ exists is exactly $d^*$. 

\vspace*{3mm}
\noindent\textbf{Minimum-overload covers.}
Fix parameters $\tau_{\ell} < \tau_{u}$.
For a cover $M$, define the \emph{overload} of a clique vertex $v\in C$ by $\mathrm{ov}_M(v) := \max\{0,\deg_M(v)-\tau_{\ell}\},$
and define the total overload as
$\mathrm{OV}(M) := \sum_{v\in C} \mathrm{ov}_M(v).$
Among all covers $M$ with $\Delta(M)\le \tau_{v}$,
let $M^{\mathrm{ov}}$ be one minimizing $\mathrm{OV}(M)$,
and let $R := \mathrm{OV}(M^{\mathrm{ov}})$. See Figure~\ref{fig:covers} for an illustration.

Among all covers of $B$, the one with the minimum-overload cover can be found via a reduction to the min-cost flow, which can be solved optimally in polynomial time (see, e.g., ~\cite{flowBook}). For the reduction, we construct a network flow from $B$ as follows:
$s'$ is connected to every vertex $u\in C$ with two parallel edges:
one of capacity $\tau_{\ell}$ and cost~$0$,
and one of capacity $\tau_{u}-\tau_{\ell}$ and cost~$1$;
each edge $(u,v)\in E(B)$ has capacity~$1$ and cost~$0$;
each vertex $v\in I$ is connected to the sink with an edge of capacity~$1$ and cost~$0$.
Sending one unit of flow through the costly edge corresponds to one unit of overload.
Minimizing the total cost yields a cover minimizing overload. \vspace{2mm}


\begin{figure}[!t]
 \begin{minipage}[b]{0.48\textwidth}
\centering        \scalebox{.71}{\coverOne}
\subcaption{Cover $M$} 

        \label{fig:coverOne}
    \end{minipage} \  \ 
    \hfill 
    \begin{minipage}[b]{0.48\textwidth}
        \centering
        \scalebox{.71}{\coverTwo} 
        \centering \subcaption{Cover $M'$} 
        \label{fig:coverTwo}
    \end{minipage} \  \   
\caption{Two minimum-degree coverings of $B$ with $d^* = 3$. When $\tau_\ell = 2$, we have $OV(M) = 2$ (both $u_1$ and $u_2$ have overload 1), while $OV(M')=1$ (only $u_1$ has overlaod 1).}
    \label{fig:covers}
\end{figure}

\noindent\textbf{The \SIMPLE algorithm.}
We begin with a warm-up approximation algorithm of~\cite{HarutyunyanHSplit23}, which we call \SIMPLE. The algorithm works as follows: 

\smallskip
\noindent
\begin{enumerate}
    \item Broadcast within the clique $C$ until all clique vertices are informed.
    \item Use a minimum-degree cover $M$ of $B$ to inform the vertices in $I$:
    each $v\in C$ informs its assigned neighbors one per round.
\end{enumerate}

Step~(1) completes in at most $\lceil \log |C| \rceil$ rounds.
Step~(2) requires $\Delta(M)=d^*$ additional rounds.
Thus, $T_{\textsc{Simple}} \le \lceil \log |C| \rceil + d^*.$ Both $\lceil \log |C| \rceil$ and $d^*$ are lower bounds on $\bropt$: the former is necessary to inform all clique vertices,
while the latter follows since each vertex of $I$ must be informed by some neighbor in $C$,
and each clique vertex can inform at most one new vertex per round.
Hence, $\SIMPLE$ is a $2$-approximation. \vspace{2mm}

\paragraph*{\splitalgo: Improved Approximation for Split Graphs}
Our approximation algorithm, named \splitalgo, refines \textsc{Simple} by
distinguishing clique vertices according to their load in
the minimum-overload cover.
Intuitively, the vertices with high overload correspond to unavoidable congestion in any broadcast schedule and are handled in a separate phase to be informed earlier, thereby improving the approximation ratio.
For that, we compute a cover that keeps the maximum load below a prescribed threshold while minimizing the extent to which the assignment exceeds a smaller ``soft cap.''
In particular, \splitalgo\ proceeds as follows (see Algorithm~\ref{alg:split}):

\begin{algorithm}[t]
\caption{\textsc{SplitBroadcast} (single source)}
\label{alg:split}
\begin{algorithmic}[1]
\REQUIRE Split graph $G=(C\cup I,E)$ with source $s$
\STATE Compute a minimum-degree cover of $B$ and let $d^*$ be its degree.
\STATE Set thresholds 
$\tau_\ell = \alpha d^*$, $\tau_m = (\alpha+\epsilon) d^*$, and $\tau_u = \beta d^*$ (for some $\alpha \leq 1 \leq \beta,$ and small $\epsilon>0$).

\STATE Compute a minimum-overload cover $M^{\mathrm{ov}}$ with cap $\tau_{u}$ and overload parameter $\tau_{\ell}$.
\STATE Let $H := \{u\in C : \deg_{M^{\mathrm{ov}}}(u)\ge \tau_m\}$ (high-load vertices) and \\ $L := C\setminus H$ (low-load vertices).
\STATE Form the following broadcast scheme \CLASS:
\begin{itemize}
    \item Ensure that at least one vertex of $H$ and one vertex of $L$ are informed (within two rounds).
    \item Broadcast within $H$ until all vertices of $H$ are informed.
    \item Broadcast from vertices of $H$ to their neighbors in $I$ assigned by $M^{\mathrm{ov}}$.
    \item In parallel, broadcast within $L$ and then from $L$ to their neighbors in $I$ assigned by $M^{\mathrm{ov}}$.
\end{itemize}
\STATE Between \SIMPLE and \CLASS, return the faster scheme.
\end{algorithmic}
\end{algorithm}

\smallskip
\noindent
\emph{Step 1:} Compute a minimum-degree cover of $B$ and let $d^*$ be its degree.\\
\emph{Step 2:} Set thresholds $\tau_\ell = \alpha d^*$, $\tau_m = (\alpha+\epsilon) d^*$, and $\tau_u = \beta d^*$ for some $\alpha \leq 1 \leq \beta,$ that will be decided later, and a small $\epsilon>0$. \\
\emph{Step 3: Compute a minimum-overload cover.}
Use the min-cost flow reduction described earlier to compute the minimum overload cover of $B$ with parameters $\tau_\ell$ and $\tau_u$ in polynomial time.

\smallskip
\noindent
\emph{Step 4: Classify clique vertices.}
Using $M^{\mathrm{ov}}$, classify clique vertices as \emph{high-load} if they are assigned at least
$\tau_m=(\alpha+\epsilon)d^*$ neighbors, and \emph{low-load} otherwise.
Let $H$ be the set of high-load vertices and $L=C\setminus H$.

\smallskip
\noindent
\emph{Step 5: Form the broadcast scheme \CLASS as follows}
\begin{itemize}
    \item \emph{Step 5(a): Seeding.}
    Within two rounds, ensure that at least one vertex in $H$ and at least one vertex in $L$ becomes informed.
It takes one round to inform one vertex $u$ in $C$ (if $s\in I$), and another round 
to inform another vertex in $C$ which belongs to a different group than $u$. 
\item{Step 5(b): Two parallel broadcast processes.}
From this point on, run two broadcast processes in parallel.
First, broadcast within the clique subgraph induced by $H$ until all vertices in $H$ are informed, and then each $u\in H$
informs its assigned independent neighbors (those $v\in I$ with $(u,v)\in M^{\mathrm{ov}}$) one per round.
This part takes at most $\lceil \log |H|\rceil + \tau_{u}$ rounds after the seeding step.

In parallel, we broadcast within the clique subgraph induced by $L$ until all vertices in $L$ are informed, and then each $u\in L$ informs its assigned independent neighbors one per round.
Since every low-load vertex has at most $\tau_m$ assigned neighbors, this part takes at most $\lceil \log |C|\rceil + \tau_m$ rounds after the seeding step.
\end{itemize}

\noindent\emph{Step 6: Among \SIMPLE and \CLASS, return the scheme that has a smaller broadcast time. }


\paragraph*{Analysis of \splitalgo.}
We analyze the broadcast time of the algorithm. 
In what follows, we let $n$ denote $|C|$; recall that $d^*$ denotes the degree of a minimum-degree cover of $B$,
and $M^{\mathrm{ov}}$ is a minimum-overload cover with parameters
$\tau_\ell=\alpha d^*$, $\tau_m=( \alpha+\epsilon)d^*$, and $\tau_u=\beta d^*$.
Let $H$ denote the set of clique vertices whose load in $M^{\mathrm{ov}}$ is at least $\tau_m$,
and let $n_h=|H|$.

\smallskip
\noindent
\textbf{Upper bound on the algorithm.}
After two seeding rounds, the algorithm \CLASS runs two broadcast processes in parallel within $H$ and $L$.
Broadcasting within $H$ takes at most $\lceil \log n_h\rceil$ rounds,
followed by at most $\tau_u$ rounds to inform the independent vertices
assigned to $H$ by $M^{\mathrm{ov}}$.
Similarly, broadcasting within $L=C\setminus H$ takes at most $\lceil \log n\rceil$
rounds, followed by at most $\tau_m$ rounds to inform their assigned neighbors.
Hence, the broadcast time of the algorithms satisfy
\begin{align}
\br({\mathrm{CLASS}})
\;& \le\;
2
+\max\!\left\{
\lceil \log n_h\rceil + \beta d^*,\;
\lceil \log n\rceil + (\alpha+\epsilon) d^*
\right\}, \\ 
\br({\text{\splitalgo}}) & = \min\{\br({\text{\SIMPLE}}), {\br(\text{\CLASS}})\}. \label{eq:zero}
\end{align}

\smallskip
\noindent
\textbf{Lower bounds for $\bropt$.}
Throughout the analysis, we make use of the following lemma, which states that there exists an optimal broadcast tree in which every vertex in $I$—except possibly the source, if it belongs to $I$—appears as a leaf.

\begin{figure}[!t]
 \begin{minipage}[b]{0.48\textwidth}
\centering        \scalebox{.75}{\figMatchConvertOne}
\subcaption{Broadcast tree before the transformation} 

        \label{fig:matchConvertOne}
    \end{minipage} \  \ 
    \hfill 
    \begin{minipage}[b]{0.48\textwidth}
        \centering
        \scalebox{.75}{\figMatchConvertTwo} 
        \centering \subcaption{Broadcast tree after the transformation}
        \label{fig:matchConvertTne}
    \end{minipage} \  \   
\caption{An example showing how a broadcast scheme with a vertex $v\in I$ as a non-leaf can be transformed, without increasing the broadcast time, to eliminate non-leaf vertices in the independent set.}
    \label{fig:matchConvert}
\end{figure}

\begin{lemma}\label{lemma:goodopt}
Let $G = (C \cup I, E)$ be a split graph with clique $C$ and independent set $I$, and let $s \in V(G)$ be the source vertex. There exists an optimal broadcasting scheme in which every non-source vertex of $I$ is a leaf in the broadcasting tree.
\end{lemma}

\begin{proof}
Let $T$ be an optimal broadcast tree. Suppose $T$ contains a vertex $v\in I$ that is a non-source internal node, and let $z$ be the \emph{first} vertex that $v$ informs. 
We show how to modify $T$, without increasing its broadcast time, so that $v$ has one fewer child. By repeating this transformation, we eventually obtain a broadcast tree in which every vertex in $I$ is either a leaf or the root, as desired.

\emph{Modifying $T$:} remove the edge $v \to z$ and instead have $u$ call $z$ at the time it originally called $v$, then call $v$ one step later. That is, $z$ takes $v$'s former slot in $u$'s calling sequence. Moreover, $z$ informs $v$ as its first child. See Figure~\ref{fig:matchConvertTne} for an illustration. We verify that the broadcast time does not increase

\begin{itemize}
    \item \textbf{$z$ and its subtree:} $z$ is informed one step \emph{earlier}, so every vertex in its subtree is informed earlier.
    \item \textbf{$v$ and its remaining children:} $v$ is informed one step later, but has lost its first child $z$. Each remaining child of $v$ (originally called at times $t_v + 2, t_v + 3, \ldots$) is now called at $(t_v+1)+1, (t_v+1)+2, \ldots$\,---\,the same times as before.
    \item Other vertices receive at the same time as before.
\end{itemize}

Therefore, the modified tree remains optimal. 
\end{proof}

In light of the above lemma, we will consider an \emph{optimal cover}, defined as a covering of $B$ induced by an optimal broadcast tree in which all vertices of $I$ appear only as leaves (or possibly the root) and are covered by their neighbors. This assumption is used in the proof of the following theorem, which is the main result of this section.

\begin{theorem}\label{Th:SplitMain}
\splitalgo gives an approximation factor of \apxsplit for \telebr in split graphs. 
\end{theorem}

\begin{proof}
As discussed before, we have $\bropt\ge \max \{\lceil \log n\rceil, d^*\}$. We use a case analysis to complete the proof.

\smallskip
\noindent
\textbf{Case 1:} 
Suppose the optimal cover has degree at least $\beta d^*$. This implies that $d^* \leq \bropt/\beta$.
In this case, \textsc{Simple} completes in
\begin{align}
\br(\SIMPLE) \leq \lceil \log n\rceil + d^*
\;\le\;
(1+1/\beta)\bropt. \label{eq:one}    
\end{align}

Since \splitalgo returns the better scheme between \SIMPLE and \CLASS, its approximation factor in this case is at most $(1+1/\beta)$.

\smallskip
\noindent
\textbf{Case 2:} Conisder otherwise. Therefore, $d^* > \bropt/\beta$.
Recall that, in the broadcast scheme of \CLASS, the last leaf with a parent in $H$ is informed no later than $\lceil \log n_h \rceil + \beta d^*$. 

In this case,
\begin{align}
 2+\lceil \log n\rceil + (\alpha+\epsilon) d^*
\;\le\;
(1+\alpha+\epsilon)\,\bropt + 2, \label{eq:two}
\end{align}


By optimality of $M^{\mathrm{ov}}$, the total overload $R$ of
$M^{\mathrm{ov}}$ with respect to $\tau_\ell$ and $\tau_u$ is at most the overload of the optimal cover.
Moreover, in the optimal cover, each clique vertex whose load is at least $\tau_\ell$ contributes at most
$\tau_u-\tau_\ell=(\beta- \alpha)d^*$ units of overload.
It follows that in the optimal cover, there are at least
$\frac{R}{(\beta-\alpha)d^*}$ clique vertices of load at least $\tau_\ell$.
Informing these vertices and their neighbors in $I$ defines a lower bound for the optimal broadcast time. That is, 
\begin{align*}
\bropt \geq \lceil \log \frac{R}{(\beta-\alpha)d^*}\rceil
+\alpha d^* \geq \log \frac{R}{d^*}
+\alpha d^* -\log (\beta-\alpha).
\end{align*}


On the other hand, every vertex in $H$ contributes at least
$\tau_m-\tau_\ell=\epsilon d^*$ units of overload in the broadcast scheme of \CLASS, implying
$n_h \le \frac{R}{\epsilon d^*}$.
Substituting this bound on $n_h$ into the upper bound on $T_{\mathrm{ALG}}$ yields
\begin{align}
2 + \lceil \log n_h \rceil + \beta d^* &\leq  
\log \frac{R}{\epsilon d^*} + \beta d^* + 3 \nonumber\\
&=\log \frac{R}{d^*} + \beta d^* + 3 -\log \epsilon  \nonumber\\
\;&\le\;
\frac{\beta}{\alpha}\,\bropt + 3 - \log \epsilon - \frac{\beta}{\alpha} \log (\beta-\alpha) \nonumber\\
&=\frac{\beta}{\alpha}\,\bropt + \oh(1). \label{eq:three}
\end{align}

Therefore, for the broadcast time of \CLASS, 
we can conclude 
\begin{align*}
\br(\CLASS) & = 2
+\max\!\left\{
\lceil \log n_h\rceil + \beta d^*,\;
\lceil \log n\rceil + (\alpha+\epsilon) d^*
\right\}  &\text{(by \ Equation~\ref{eq:zero})} \\
& \leq \max\{(1+\alpha+\epsilon) \bropt, \frac{\beta}{\alpha} \bropt + \oh(1)\} &\text{(by \ Equations~\ref{eq:two} and \ref{eq:three})}
\end{align*}

To conclude, the broadcast time of \splitalgo is at most $(1+1/\beta)\bropt$ in Case 1 (Equation~\ref{eq:two}) and at most $\max\{ (1+\alpha+\epsilon) \bropt, \frac{\beta}{\alpha} \bropt + \oh(1)\}$ in Case 2. This gives an approximation factor of $\max\{1+1/\beta,1+\alpha+\epsilon, \beta/\alpha \}$ for \splitalgo. For $\alpha =0.754 , \beta = 1/\alpha,$ and small $\epsilon$, this would be come $1/\alpha^2 \approx \apxsplit$.
\end{proof}

\subsection*{Extension to multi-source telephone broadcasting}

In the multi-source setting, we first ensure that a set $U$ of clique vertices—consisting of clique sources and clique vertices matched to sources in the independent set—is informed in the first round.
We then restrict attention to the remaining bipartite graph and apply the same covering-based strategy as in the single-source case.
The rest of the algorithm and its analysis carry over verbatim, with the number of initially informed clique vertices $|U|$ effectively reducing the broadcast depth by $\lfloor \log |U| \rfloor$.

We now describe the extension in more detail.
Let $I_s$ and $C_s$ denote the sets of source vertices in $I$ and $C$, respectively.
If $I_s$ is nonempty, we first compute a maximal matching between $I_s$ and vertices in $C\setminus C_s$.
Let $U\subseteq C$ be the set of clique vertices that are either sources themselves or matched to a source in $I_s$.
Consequently, all vertices in $U$ are informed by the end of round~$1$.

Let $B'$ denote the bipartite graph induced by the vertex sets $C$ and $I\setminus I_s$.
After round~1, the vertices in $U$ have the message and must inform the remaining uninformed vertices in $C\cup I$.
From this point onward, the algorithm proceeds analogously to the single-source setting.
In particular, we compute the minimum-degree cover and the minimum-overload cover with respect to $B'$.
The parameters $\alpha$ and $\beta$ are defined in the same way as before, and \splitalgo selects the better of the two strategies \SIMPLE and \CLASS.

The procedures \SIMPLE and \CLASS are also defined analogously.
The only modification is that, in the second round, \CLASS ensures that at least $|U|$ vertices from each of the sets $L$ and $H$ are informed.
Subsequently, vertices in $L$ and $H$ are informed in parallel, requiring at most $\lceil \log(n/|U|)\rceil$ and $\lceil \log(n_H/|U|)\rceil$ additional rounds, respectively.

The analysis is identical to the single-source case, with the only difference that $n$ is replaced by $n/|U|$ in the cost analysis of both \splitalgo and the optimal broadcast scheme.




    
\section{Concluding Remarks}



In this work, we advanced the study of \telebr by establishing new hardness results and developing approximation schemes for specific graph families. On the hardness side, we proved NP-hardness for very restricted families of \kcycle and \kpath graphs. 
On the algorithmic side, we presented efficient polynomial-time approximation schemes (EPTASs) for both \kcycle and \kpath graphs; to the best of our knowledge, these are the first PTAS-type results for \telebr.
We also introduced an optimal, polynomial-time algorithm for graphs of bounded cutwidth, thereby extending the tractable frontier of the problem.
Finally, we developed a polynomial-time approximation algorithm for split graphs that improves upon the previously known factor--2. 

A possible direction for future work is to improve the
 approximation factor for split graphs. In particular, it
remains an open question whether Polynomial Time Approximation Schemes (PTASs) exist for these graph classes. The major open problem in this domain is determining whether a constant-factor approximation exists for general graphs. While progress on this question has been slow, the algorithmic ideas developed in this work may be applicable to broadcasting in
other families of sparse graphs. For example, it would be interesting to adapt the integer covering formulations to other broadcasting models or broader sparse graph classes. 
For series–parallel graphs, our hardness result via \kpath graphs rules out tractability, yet, to the best of our knowledge, no constant-factor approximation is currently known.




 \pagebreak
\bibliography{refstidy.bib}
\bibliographystyle{plainnat}





    


    
    

\end{document}